\documentclass[letterpaper,11pt]{article} 
\usepackage[letterpaper,hmargin=1in,vmargin=1in]{geometry} 
\usepackage{fullpage}
\usepackage[pdfstartview=FitH,colorlinks,linkcolor=blue,citecolor=blue]{hyperref} 

\usepackage[T1]{fontenc}
\usepackage{color} 
\usepackage{times} 
\usepackage{microtype,pdfsync} 
\usepackage{amsmath,amsfonts,amssymb}
\usepackage{amsthm} 
\usepackage{url} 
\usepackage{xspace,paralist}
\usepackage{lmodern}
\usepackage{float}
\floatstyle{boxed}
\restylefloat{figure}
\usepackage{multirow}
\usepackage{makecell}
\usepackage{ stmaryrd }

\usepackage{marvosym}
\usepackage{graphicx}
\usepackage{graphbox}
\usepackage{subcaption}
\usepackage{caption}
\usepackage{ bbold }
\usepackage{breakcites}
\usepackage{mathtools}

\usepackage{xfrac}

\usepackage{pifont}

\theoremstyle{plain} 

\newtheorem{theorem}{Theorem}
\newtheorem{lemma}[theorem]{Lemma} 
\newtheorem{corollary}[theorem]{Corollary}
\newtheorem{claim}[theorem]{Claim}
\newtheorem{fact}[theorem]{Fact}
\newtheorem{remark}{Remark} 

\theoremstyle{definition} 

\newtheorem{definition}[theorem]{Definition}

\numberwithin{theorem}{section}
\numberwithin{equation}{section}

\newcommand{\Z}{\mathbb{Z}}

\newcommand{\cv}{{\mathbf{c}}}

\newcommand{\Id}{{\mathbf{I}}}
\newcommand{\Am}{{\mathbf{A}}}

\newcommand{\Hm}{{\mathbf{H}}}

\newcommand{\Om}{{\mathbf{O}}}
\newcommand{\Pm}{{\mathbf{P}}}

\newcommand{\Um}{{\mathbf{U}}}
\newcommand{\Vm}{{\mathbf{V}}}

\newcommand{\Tr}{{\sf Tr}}

\newcommand{\As}{{\mathcal{A}}}
\newcommand{\Bs}{{\mathcal{B}}}
\newcommand{\Cs}{{\mathcal{C}}}
\newcommand{\Ds}{{\mathcal{D}}}

\newcommand{\Fs}{{\mathcal{F}}}
\newcommand{\Gs}{{\mathcal{G}}}
\newcommand{\Hs}{{\mathcal{H}}}

\newcommand{\Ms}{{\mathcal{M}}}

\newcommand{\Os}{{\mathcal{O}}}
\newcommand{\Ps}{{\mathcal{P}}}
\newcommand{\Qs}{{\mathcal{Q}}}
\newcommand{\Rs}{{\mathcal{R}}}
\newcommand{\Ss}{{\mathcal{S}}}
\newcommand{\Ts}{{\mathcal{T}}}
\newcommand{\Us}{{\mathcal{U}}}

\newcommand{\Xs}{{\mathcal{X}}}

\newcommand{\poly}{{\sf poly}}
\newcommand{\polylog}{{\sf polylog}}

\newcommand{\setupprotocol}[2]{%
	\expandafter\newcommand\csname protocol#1\endcsname{{\sf \Pi_{#2}}}%
	\expandafter\newcommand\csname gen#1\endcsname{{\sf Gen_{#2}}}%
	\expandafter\newcommand\csname genlossy#1\endcsname{{\sf GenLossy_{#2}}}%
	\expandafter\newcommand\csname setup#1\endcsname{{\sf Setup_{#2}}}%
	\expandafter\newcommand\csname enc#1\endcsname{{\sf Enc_{#2}}}%
	\expandafter\newcommand\csname dec#1\endcsname{{\sf Dec_{#2}}}%
	\expandafter\newcommand\csname extract#1\endcsname{{\sf Extract_{#2}}}%
	\expandafter\newcommand\csname derive#1\endcsname{{\sf Derive_{#2}}}%
	\expandafter\newcommand\csname embed#1\endcsname{{\sf Embed_{#2}}}%
	\expandafter\newcommand\csname rerand#1\endcsname{{\sf ReRand_{#2}}}%
	\expandafter\newcommand\csname trace#1\endcsname{{\sf Trace_{#2}}}%
	\expandafter\newcommand\csname goodtrace#1\endcsname{{\sf GoodTr_{#2}}}%
	\expandafter\newcommand\csname badtrace#1\endcsname{{\sf BadTr_{#2}}}%
	\expandafter\newcommand\csname gooddecoder#1\endcsname{{\sf GoodDec_{#2}}}%
	\expandafter\newcommand\csname goodf#1\endcsname{{\sf Goodfunc_{#2}}}%
	\expandafter\newcommand\csname badextract#1\endcsname{{\sf BadExtr_{#2}}}%
	\expandafter\newcommand\csname adv#1\endcsname{{\As_{#2}}}%
	\expandafter\newcommand\csname sample#1\endcsname{{\sf Sample_{#2}}}%
	\expandafter\newcommand\csname diff#1\endcsname{{\sf Diff_{#2}}}%
	\expandafter\newcommand\csname identify#1\endcsname{{\sf Indentify_{#2}}}%
	\expandafter\newcommand\csname findtags#1\endcsname{{\sf FindTags_{#2}}}%
	\expandafter\newcommand\csname confirmtags#1\endcsname{{\sf ConfirmTags_{#2}}}%
	\expandafter\newcommand\csname eval#1\endcsname{{\sf Eval_{#2}}}%
	\expandafter\newcommand\csname sign#1\endcsname{{\sf Sign_{#2}}}%
	\expandafter\newcommand\csname ver#1\endcsname{{\sf Ver_{#2}}}%
	\expandafter\newcommand\csname prf#1\endcsname{{\sf F_{#2}}}%
	\expandafter\newcommand\csname Sim#1\endcsname{{\sf Sim_{#2}}}%

	\expandafter\newcommand\csname msk#1\endcsname{{\sf msk_{#2}}}%
	\expandafter\newcommand\csname sk#1\endcsname{{\sf sk_{#2}}}%
	\expandafter\newcommand\csname pk#1\endcsname{{\sf pk_{#2}}}%
}

\setupprotocol{}{}

\newcommand{\out}{{\sf out}}

\newcommand{\SWAP}{{{\sf SWAP}}}
\newcommand{\Sym}{{{\sf Sym}}}
\newcommand{\Phase}{{{\sf Ph}}}

\newcommand{\QFT}{{{\sf QFT}}}

\def\E{\mathop{\mathbb E}}

\newcommand{\ignore}[1]{}

\floatstyle{plain}
\restylefloat{figure}

\begin{document}

\title{The Space-Time Cost of Purifying Quantum Computations}
\author{Mark Zhandry\\NTT Research}
\date{}

\maketitle

\begin{abstract}General quantum computation consists of unitary operations and also measurements. It is well known that intermediate quantum measurements can be deferred to the end of the computation, resulting in an equivalent purely unitary computation. While time efficient, this transformation blows up the space to linear in the running time, which could be super-polynomial for low-space algorithms. Fefferman and Remscrim (STOC'21) and Girish, Raz and Zhan (ICALP'21) show different transformations which are space efficient, but blow up the running time by a factor that is exponential in the space. This leaves the case of algorithms with small-but-super-logarithmic space as incurring a large blowup in either time or space complexity. We show that such a blowup is likely inherent, demonstrating that any ``black-box'' transformation which removes intermediate measurements must significantly blow up either space or time.
\end{abstract}
	
\section{Introduction}\label{sec:intro}

Measurements play a fundamental role in quantum computation. After all, it is through measurements that useful classical information is extracted from the hidden world of a quantum state. That said, formal treatments typically regard a quantum computation as being \emph{unitary}, with any measurement only occurring at the very end of the computation. For example, many algorithmic techniques such as amplitude amplification~\cite{ISTCS:BraHoy97,Grover97,BHMT02}, numerous query complexity lower bounds techniques~\cite{BBBV97,FOCS:BBCMW98,STOC:Ambainis00}, and cryptographic proofs involving rewinding~\cite{Watrous09,EC:Unruh12,EC:KSSSS20,FOCS:CMSZ21,FOCS:LomMaSpo22}, all assume unitary algorithms whose states are pure. On the other hand, when designing quantum algorithms it is often convenient to measure and/or discard quantum states in the middle of a computation. Unitary computations may also be desireable from a practical perspective, as implementing measurements can be challenging, and may have energy-use implications (see Section~\ref{sec:reverse} below).
Fortunately, assuming unitary computations can be justified by appealing to the ``Principle of Delayed Measurement,'' which states that measurements in general quantum computations can always be delayed until the end of the computation with minimal time-complexity overhead. This is accomplished by, instead of measuring the qubit, writing it into an external register that is never touched again in the computation.

However, it has long been recognized that delaying measurements naively gives a space complexity that is potentially as large as the time complexity, even if the original computation used very little space. Thus, delaying measurements potentially incurs a huge space-complexity overhead. Eliminating measurements in a space-efficient way has therefore become a major foundational question in quantum computation. This question may also be of practical importance, as quantum storage will plausibly be a limiting resource in future quantum computers.

Fefferman and Remscrim~\cite{STOC:FefRem21} and Girish, Raz and Zhan~\cite{ICALP:GirRazZha21} give space-optimal answers to this problem, showing that intermediate measurements can be eliminated to yield a unitary computation with only a linear blowup in space complexity. However, these results incur a potentially \emph{exponential} blowup in time complexity: the new running time is $\poly(T,2^S)$ where $T,S$ are the original time and space complexities. This leaves the following important open problem:

\begin{center}{\it Can intermediate measurements be eliminated in a\\ simultaneously space- \emph{and} time-efficient manner?}\end{center}

\noindent {\bf Our work.} Our main result is to show a \emph{black-box barrier} to achieving such a result.

\subsection{What is a Quantum Measurement, Anyway?}

Before proceeding, we must mention the work of Girish and Raz~\cite{ITCS:GirRaz22}, which eliminates  intermediate ``measurements'' from any space $S$, time $T$ quantum algorithm, resulting in a space $O(S\log T)$, time $T\times\poly(S)$ algorithm without measurements. This seemingly resolves the central question above positively. However, we note that their result only works for a very particular notion of measurement.

Digging deeper, their model of computation allows for unitary gates, plus a probabilistic measurement gate defined as mapping 
\[\alpha|0\rangle+\beta|1\rangle\mapsto\begin{cases}|0\rangle&\text{ with probability }|\alpha|^2\\|1\rangle&\text{ with probability }|\beta|^2\end{cases}\enspace .\]
Crucially, the measurement gates in~\cite{ITCS:GirRaz22} do not output the classical measurement result, and their model does not allow the resulting quantum register to be discarded or reset to a fixed state.

Such a measurement gate as considered in~\cite{ITCS:GirRaz22} is \emph{unital}, meaning it maps the totally mixed state to the totally mixed state of the same dimension. Unitary operations are also unital, as is any combination of unital gates. As such, their model of quantum computation with measurements only captures unital computations. 

Not all works treat measurements in this way, and many algorithms in the literature are not described using such unital measurements. In fact, measurements are most often depicted as producing a \emph{classical} output, sometimes consuming the quantum state (such as with the POVM formalism) and sometimes leaving behind a ``collapsed'' quantum register (such as with the projective measurement formalism). A key distinguishing feature of classical information is that it can be erased, something which is forbidden with unital gates. One can also consider ``reset'' gates which reset a qubit to $|0\rangle$, or even ``discard'' gates, which simply discards a register. A depiction of some different kinds of measurement gates is given in Figure~\ref{fig:measurementgates}.

\begin{figure}[htb]
	\centering
	\begin{minipage}{0.8\textwidth}
	\mbox{\includegraphics[align=t,width=.19\textwidth]{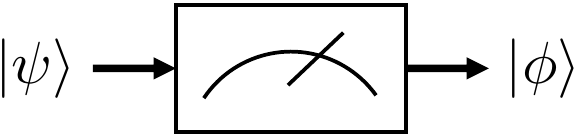}
		\hspace{.02\textwidth}\vrule\hspace{0.02\textwidth}
		\includegraphics[align=t,width=.14\textwidth]{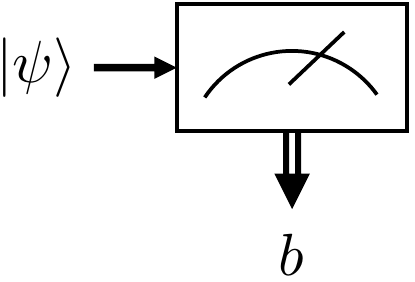}
		\hspace{.04\textwidth}
		\includegraphics[align=t,width=.19\textwidth]{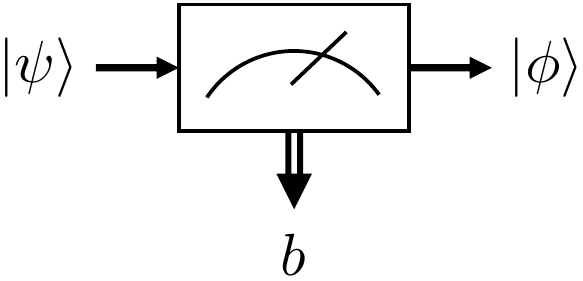}
		\hspace{.04\textwidth}
		\includegraphics[align=t,width=.10\textwidth]{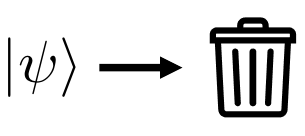}
		\hspace{.04\textwidth}
		\includegraphics[align=t,width=.14\textwidth]{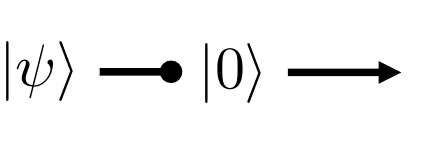}}\end{minipage}
	\captionsetup{width=.8\textwidth}
	\caption{\label{fig:measurementgates} Different kinds of measurement gates. Here, $|\psi\rangle$ is the state being measured, $b$ is the probabilistic measurement outcome, and $|\phi\rangle$ is the state that $|\psi\rangle$ collapses to when the measurement outcome is $b$.}
\end{figure}

We note that the rightmost four types of measurement gates in Figure~\ref{fig:measurementgates} --- that is, consuming the quantum state but outputing the classical measurement result, outputting both the measurement result and the collapsed state, resetting a qubit to $|0\rangle$, or simply discarding the state --- are all easily seen to be equivalent under appropriate assumptions\footnote{Assuming the ability to (1) arbitrarily discard classical values, (2) have quantum gates depend on previously obtained classical values, and (3) initialize new registers.}. In fact, unitary operations plus any one of these gates can be used to implement any quantum channel, a consequence of Stinespring Dilation~\cite{Stinespring55}. Moreover, all four appear frequently at least implicitly throughout the literature in the descriptions of quantum algorithms. On the other hand, the left-most type of gate --- the one considered in~\cite{ITCS:GirRaz22} --- which outputs the collapsed state but no classical output, is unital, meaning it alone is not enough to lift unitary operations to general channels. Thus, we see that~\cite{ITCS:GirRaz22} only applies to a version of measurement that is potentially quite limited.

More generally, one can consider a quantum computation involving general non-unitary channel gates, of which measurements are only a specific example. The goal is then to ``purify'' the computation, turning it into a computation involving only unitary gates.~\cite{ITCS:GirRaz22} will fail on general channels. We note that, in contrast to~\cite{ITCS:GirRaz22},~\cite{STOC:FefRem21} applies to quantum algorithms comprising arbitrary (potentially non-unital) channels, at the cost of a potentially exponential blowup in time complexity\footnote{\cite{ICALP:GirRazZha21}, on the other hand, only claims to apply to unital channels.}.

\begin{remark}We stress that~\cite{ITCS:GirRaz22} only claim their results work for their notion of measurement gates. They also mention that with qubit reset gates, it is trivial to simulate an intermediate measurement. But then the resulting circuit would have qubit reset gates. Qubit reset could reasonably itself be considered a ``measurement'' in a broader sense, since it is non-unitary and is equivalent to various other versions of measurements. To try to avoid any confusion, we will use the term ``general quantum computation'' to refer to computations involving this more general view of measurement.
\end{remark}

\subsection{Relationship to Classical Reversible Computation}\label{sec:reverse}

The task of eliminating intermediate measurements has an analog in classical computation: namely, turning general (irreversible) classical computation into reversible computation. One motivation for reversible computing is Landauer's principle~\cite{Landauer61}, which states that any irreversible logic operation requires a certain minimum energy consumption, therefore imposing a limit on how much efficiency can be improved. Meanwhile, no such energy consumption is inherent to reversible operations, meaning in principle reversible computation could have zero energy cost. In the quantum setting, measurements make a quantum algorithm irreversible and Landauer's principle would likewise impose a minimal energy consumption. Meanwhile, unitary algorithms are reversible and therefore ``immune'' to Landauer's principle.

Analogous to the quantum setting, in the classical setting one can make a general computation reversible trivially by blowing up the space to be linear in the running time. An old classical question was whether anything better is doable.

Bennett~\cite{Bennett89} resolved this classical question, showing that space $S$ and time $T$ general computation can be made reversible with space $S'=O(S\log T)$ and time $T'=\poly(T)$, thus preserving time and space efficiency. One may therefore be tempted to apply similar techniques to obtain an analogous result for eliminating quantum measurements. We now explain, however, that this strategy fails.

Bennett's result works roughly as follows. We first start with the trivial conversion, which makes an irreversible computation reversible by simply storing the complete program trace containing all prior states of the algorithm. To reverse a step of the computation, one simply un-computes the last state in the trace by re-computing it from the penultimate state. This of course blows up the space from $S$ to $S\times T$. What Bennett does is cleverly store only a few carefully selected prior states at a time, and show that this is sufficient to reversibly simulate the original computation, with only a modest blow-up in time complexity.

One may be tempted to adapt this technique in order to make a quantum computation with measurements reversible in low space and time, thereby removing intermediate measurements. We observe, however, that Bennett's result relies on a crucial feature of classical information that is no longer true quantumly: that the classical intermediate states of the algorithm can be copied --- one copy going into the program trace, and another copy to continue the computation. In contrast, the intermediate quantum states of a general quantum algorithm will not be copy-able by the no-cloning theorem. Of course, instead of copying, we could try to produce two copies of the state by running the algorithm a second time from the beginning. This will potentially fail for non-unitary algorithms, however, as intermediate measurements may have made the intermediate states unpredictable. But even worse, running the algorithm a second time will involve its own intermediate measurements that will need to be eliminated. So it is not clear if copying the state by running the algorithm from scratch a second time resulted in any progress. No-cloning thus seems to invalidate this approach to eliminating quantum measurements.

\subsection{Our Results}

\paragraph{Formalizing Black Box Impossibilities.} Our goal is to show that there is no procedure to eliminate general quantum measurements without blowing up either space or time. However, we observe that an unconditional result is out of reach given the current state of complexity theory. Indeed, if ${\sf BQL}={\sf BQP}$ (quantum log-space equals quantum polynomial-time), then for any ${\sf BQP}$ computation, we can trivially eliminate measurements using delayed measurements, thereby blowing up the space, but then ``compress'' the space using the equivalence to ${\sf BQL}$.

We therefore provide a notion of ``black-box'' compilers for quantum circuits. Such black-box compilers capture natural techniques such as all the quantum compilers mentioned above~\cite{STOC:FefRem21,ICALP:GirRazZha21,ITCS:GirRaz22} and the Principle of Delayed Measurement. The classical version of our notion also captures~\cite{Bennett89}, and therefore our notion of quantum black-box compiler captures natural attempts to adapt~\cite{Bennett89} to the quantum setting. We note that our notion of black-box is somewhat different than notions studied in cryptography~\cite{STOC:ImpRud89}. Indeed, all the compilers mentioned above inherently operate on the computation at the circuit level, meaning the compilers get to ``see'' the circuit representation. In contrast, black-box techniques in cryptography treat the inputs as a monolithic computation, and the techniques are explicitly forbidden from seeing the circuit representation. Our key insight is that natural circuit compilers like those discussed above do make use of the circuit representation, but are essentially agnostic to the gates used in the original computation, giving equally good space- and time-bounds regardless of the gate set used. We therefore define black-box compilers, roughly, as those that work equally well for \emph{any} set of gates.

\paragraph{Our Main Theorem.} We can now state our main theorem:

\begin{theorem}[Informal]\label{thm:maininf} For any black-box compiler mapping space $S$, time $T$ general quantum computation to space $S'$, time $T'$ unitary computation, either $S'=\Omega(T)$ or $T'=2^{\Omega(S)}$.
\end{theorem}
We note that the Principle of Delayed Measurement and~\cite{STOC:FefRem21} demonstrate that Theorem~\ref{thm:maininf} is essentially tight\footnote{Assuming the typical parameter setting where $T\leq 2^{O(S)}$.}. We prove our theorem by exhibiting, for any $S,T$, a set of unitary gates and a space $S$, time $T$ general quantum computation (with measurements) relative to these gates, such that any unitary simulation using these gates requires space $S'=\Omega(T)$ or time $T'=2^{\Omega(S)}$. Theorem~\ref{thm:maininf} also demonstrates that~\cite{ITCS:GirRaz22} cannot be generalized to handle arbitrary measurement gates.


\paragraph{New Space Lower Bound Technique.} In order to prove Theorem~\ref{thm:maininf}, we need a lower-bound technique that works on unitary computation, but crucially fails to lower bound general quantum computation containing measurements. After all, Theorem~\ref{thm:maininf} requires the existence of a low-space general quantum algorithm with measurements for the task. Prior quantum space lower bounds (e.g.~\cite{SIAM:KlaSpadeW07,QIC:NABT15,FOCS:CGLQ20,HamMag23}) typically work similarly well for both general quantum algorithms and those that make no measurements. Indeed, this would be considered a \emph{feature} in the usual setting of space lower bounds as it makes them more general. But for us, it means we need a fundamentally new lower-bound technique.

Our lower bound technique works by simulating a quantum gate using a stateful simulator. The simulator will start out having some space. Then we show that for any algorithm solving some task, the size of the simulator's state must decrease by a certain amount. As the joint state of the simulator and any algorithm that does not make measurements is pure, we argue that the total joint state size must not decrease from its original value. But since the simulator's state decreased in size, this means the algorithm's state size increased. Observe that this technique does not apply to algorithms which may make measurements, as such measurements result in a joint operation that is non-unitary and can decrease in size.


\section{Technical Overview}

\paragraph{Our Construction.} We now give an overview of our results and techniques. Motivated by the challenges of adapting~\cite{Bennett89} to remove quantum measurements, our idea is to design a computation where intermediate states are unclonable. In our case, the intermediate states are predictable. But only part of the intermediate state, call it $|\psi\rangle$, is useful, and the other part, say $|\phi\rangle$, is a useless byproduct of the computation of $|\psi\rangle$. If measurements are allowed --- or more precisely, the ability to reset registers --- then $|\phi\rangle$ can always be reset before proceeding. But if we demand a unitary version of the computation, the only way to eliminate $|\phi\rangle$ appears to be to un-compute it. But in principle, the only way to uncompute $|\phi\rangle$ is to actually undo the joint computation of $|\psi\rangle$ and $|\phi\rangle$. Remember that $|\psi\rangle$ is unclonable; this means uncomputing $|\psi\rangle$ actually just returned to a previous point in the computation, and we have not actually made any progress. If one tries to compute $|\psi\rangle$ a second time from scratch, this will work, but now there are two $|\phi\rangle$ states that need to be un-computed. There seems to be no unitary way of computing $|\psi\rangle$ without also having $|\phi\rangle$ be present. By having many intermediate steps produce useless side states $|\phi_1\rangle,|\phi_2\rangle,\cdots$ that all must be present to make progress, we force any unitary version of the computation to be large. Meanwhile, with measurements we can simply reset all the $|\phi_i\rangle$ as they are computed to re-use their space, keeping the overall space small. This is depicted in Figure~\ref{fig:constr}.

\begin{figure}[htb]
	\centering
	\includegraphics[width=.8\textwidth]{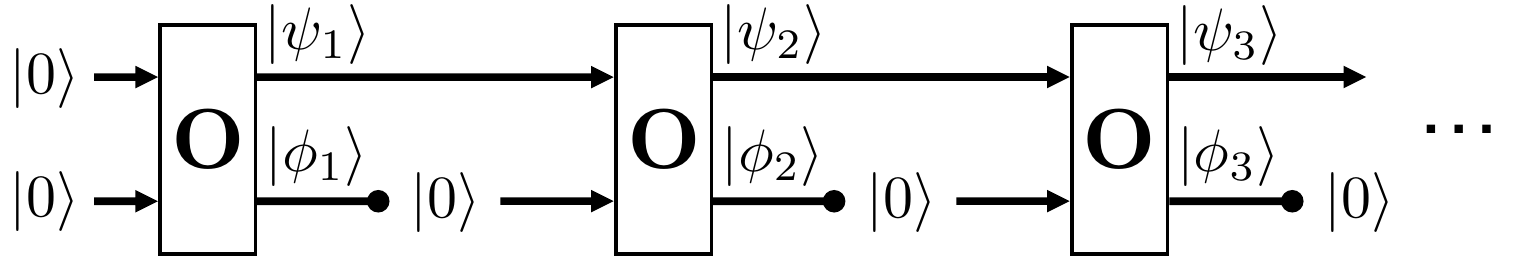}
	\captionsetup{width=.8\textwidth}
	\caption{\label{fig:constr} Our task that can be computed in low time and space with measurements, but requires large space or time without.}
\end{figure}

The final step of the computation will be to output $|\psi_t\rangle$ for some $t$, or potentially some other value that requires first computing $|\psi_t\rangle$. We observe that an algorithm that can make general quantum measurements can easily compute $|\psi_t\rangle$ in low space by iteratively computing $|\psi_i\rangle|\phi_i\rangle$ and discarding $|\phi_i\rangle$.

\paragraph{Formalizing Black Box Compilers.} As discussed earlier, it is consistent with current knowledge (even if considered unlikely) that ${\sf BQL}={\sf BQP}$, in which case one can eliminate measurements in a space- and time-efficient manner by first blowing up the space using delayed measurements, and then generically reducing the space back. However, such a mechanism would be non-black-box, in the sense that it would have to inherently use the circuit representation of the unitary $\Om$, the computation that jointly computes $|\psi\rangle,|\phi\rangle$.

We therefore imagine a class of black-box compilers, which work regardless of $\Om$. That is, any such compiler takes as input a circuit $C$ involving $\Om$ gates and measurements, and produces a new unitary circuit $C'$ using $\Om$ (and potentially other ordinary unitary gates) but no measurement gates. $C'$ must have (approximately) the same functionality as $C$. The compiler must work for \emph{any} unitary $\Om$, though we allow the compiler to have complete knowledge of $\Om$ and potentially have the choice of circuit $C'$ depend on $\Om$. The aforementioned compilers for eliminating measurements such as delayed measurements,~\cite{STOC:FefRem21,ICALP:GirRazZha21,ITCS:GirRaz22}, or any strategy similar to~\cite{Bennett89} are all black box in this sense. We explain in slightly more detail how our notion of black-box captures these works in Section~\ref{sec:compilers}.

By treating $\Om$ as a black box, we have now turned a potentially intractable problem involving at a minimum quantum complexity lower-bounds into an oracle problem, which may be tractable.

\begin{remark}Note that one can use the space-efficient version of the Solovay-Kitaev Theorem (\cite{MelWat12} Theorem 7) to replace any constant-sized set of unitary gates with any other constant-sized set of unitary gates in a space- and time-efficient manner. However, this transformation is only efficient when fixing the gate sets and then considering the complexities asymptotically; the constants in the asymptotics will depend on the gate sets in question. In particular, if we let $n$ be the number of qubits $\Om$ acts on, applying~\cite{MelWat12} to replace $\Om$ with gates from a fixed universal gate set will blow up the running time to $2^{\Omega(n)}$. This is ``constant'' if $\Om$ and $n$ are fixed, but is exponential if we allow $n$ to vary. In our case, we set $n=S$, the space of $C$, in which case applying~\cite{MelWat12} gives a running time of at least $2^{\Omega(S)}$.
	
Our notion of a black-box compiler requires the space and time complexities $S',T'$ of $C'$ to be fixed functions $S'=S'(S,T),T'=T'(S,T)$ of the space and time complexities $S,T$ of $C$. The functions $S',T'$ have to be the same, regardless of the gate set used by $C$ or how $C$ is constructed. We stress, however, that our notion allows $C'$ to depend arbitrarily on $C$ and its gate set, with the only restriction being on the space and time complexities. Restricting the space and time complexities in this way seems inherent: if $S',T'$ as functions of $S,T$ were allowed to depend on the gate set, then we can apply the space-efficient Solovay-Kitaev Theorem to move to a fixed universal gate set. Then, if ${\sf BQL}={\sf BQP}$, we can eliminate intermediate measurements in low space and time as explained above.
\end{remark}

\paragraph{Proving Large Unitary Space.} We now turn to proving that any unitary computation which computes $|\psi_t\rangle$ efficiently must have large space, specifically we prove that $\Omega(t\times S)$ unitary space is necessary, where $S$ is the space of the non-unitary computation in Figure~\ref{fig:constr} and is also proportional to the size of $|\phi_t\rangle$. Note that our modeling has the size of the non-unitary circuit be $T:=t\times S$, as resetting $\Theta(S)$ qubits in our model requires $\Theta(S)$ gates. Our $\Omega(t\times S)$ unitary space lower bound then implies that any efficient \emph{unitary} computation of $|\psi_t\rangle$ requires space $\Omega(T)$, matching what one gets via delayed measurements. 

More specifically, our goal is to show, roughly, that the only way to compute $|\psi_t\rangle$ efficiently with unitaries requires computing and storing each of $|\phi_1\rangle,\cdots,|\phi_t\rangle$. The challenge is, of course, that the algorithm can apply arbitrary unitaries to the $|\phi_i\rangle$, including applications of $\Om$. So we cannot hope to say that each of the $|\phi_i\rangle$ are explicitly stored in memory, as they may be hidden behind a more complex computation.

Another challenge, as mentioned previously, is that existing quantum space lower bounds make no distinction between unitary and non-unitary algorithms. Since we have a low-space non-unitary computation, any attempt to use existing techniques would necessarily fail at giving meaningful unitary lower bounds.

\paragraph{Simulating $\Om$ Statefully.} We show how to simulate the oracle $\Om$. Our simulator will only use several copies of each of $|\phi_1\rangle,\cdots,|\phi_t\rangle$ and $|\psi_1\rangle,\cdots,|\psi_t\rangle$. Essentially, whenever the gate must output $|\psi_i\rangle|\phi_i\rangle$, instead of constructing the state our simulator will simply swap in one of its copies, thereby reducing the number of copies the simulator has. Likewise, whenever the gate takes as input $|\psi_i\rangle$ and must eliminate it by uncomputing it, the simulator instead moves $|\psi_i\rangle$ from the algorithm's registers to the simulator's list of copies.

Importantly, as any supposed algorithm makes progress towards computing $|\psi_t\rangle$, we show that if the states are Haar random, then the number of copies of the various $|\phi_i\rangle$ the simulator has must \emph{decrease}. But if the algorithm is unitary, the overall joint state size can never decrease, since the initial copies of the various $|\phi_i\rangle$ cannot be destroyed by unitary computation. Therefore, if the simulator's storage decreases, the algorithm's storage must increase. Observe that this space bound does \emph{not} apply to algorithms with measurements, which can easily destroy copies of $|\phi_i\rangle$ by measuring/resetting them. This means the joint system of the non-unitary algorithm and simulator could decrease in space. Indeed, this is what happens in our low-space measurement-based algorithm.

\begin{remark}Our arguments above only apply to algorithms with running time at most $2^{O(S)}$; this is inherent as our black-box notion captures the low-space algorithm of~\cite{STOC:FefRem21}, which runs in time $2^{\Omega(S)}$. The restriction to running time $2^{O(S)}$ appears in two places. First, the claim that the number of copies of a state cannot be unitarily changed only holds for a bounded number of copies, since beyond $2^{O(S)}$ copies it is possible to perform tomography on the state. Our simulator must have a number of copies that is at least the number of queries made to $\Om$, so our arguments only apply if the number of queries is bounded. The second place where we assume a bounded running time is that our simulation introduces a small error of order $2^{-O(S)}$ for each query to $\Om$, and after $2^{O(S)}$ queries the error becomes $O(1)$, meaning the simulation failed. 
\end{remark}

\section{Preliminaries}\label{sec:prelim}

\paragraph{Quantum Computation.} A quantum system is associated with a finite-dimensional complex Hilbert space $\Hs$. A (pure) state over a quantum system is a unit column vector $|\psi\rangle$ with $\||\psi\rangle\|=1$. The Hermitian transpose of $|\psi\rangle$ is denoted $\langle\psi|$. A probability distribution over pure states is a mixed state, and is characterized by its density matrix $\rho=\sum_i p_i|\psi_i\rangle\langle\psi_i|$ where $p_i$ is the probability of $|\psi_i\rangle$. Note that $\Tr(\rho)=1$. When the distribution over $i$ is clear, we can also write $\rho=\E_i[|\psi_i\rangle\langle\psi_i|]$.

Given a complex matrix $\Um$, let $\Um^\dagger$ be the Hermitian transpose. A unitary operation is a complex square matrix $\Um$ such that $\Um\Um^\dagger=\Id$. Unitary evolution of a quantum system is described by a unitary $\Um$ that transforms $|\psi\rangle$ into $\Um|\psi\rangle$.

General \emph{non-unitary} evolution of a quantum system is described by a completely-positive
trace-preserving (CPTP) map $M$ from system $\Hs_{\sf in}$ to $\Hs_{\sf out}$. Such maps are in particular linear on density matrices, and trace preserving: $\Tr(M(\rho))=\Tr(\rho)=1$. Given a joint system $\As\otimes\Bs$, a special CPTP map is the partial trace $\Tr_\Bs$ which maps $\As\otimes\Bs$ to $\As$, with the property that $\Tr_\Bs(\rho_\As\otimes\rho_\Bs)=\rho_\As$ for any mixed states $\rho_\As,\rho_\Bs$ over $\As,\Bs$ respectively. By linearity, $\Tr_\Bs$ can be extended to all mixed state inputs. If we apply $\Tr_\Bs$ to a quantum state, we will say that $\Bs$ has been traced out.

Given any unitary operation $\Um$ on $\As$, we can extend it to a unitary operation $\Um\otimes\Id$ on $\As\otimes\Bs$ by acting as the identity on $\Bs$. Likewise, a CPTP map $M$ from $\As$ to $\As'$ can be extended to a CPTP map from $\As\otimes\Bs$ to $\As'\otimes\Bs$ by acting as the identity on $\Bs$. In both cases, we will abuse terminology and say that $\Um$ or $M$ is acting on $\As\otimes\Bs$.

\paragraph{Distance.} The \emph{trace distance} between two mixed quantum states $\rho,\rho'$ over the same system $\Hs$ is given by $TD(\rho,\rho')=\frac{1}{2}\Tr(\sqrt{(\rho-\rho')\cdot(\rho-\rho')})$. The trace distance is equivalent to the optimal distinguishing probability between the two states. The distance between two distributions $D,D'$, denoted $\Delta(D,D')$, is given by $\frac{1}{2}\sum_x |\Pr[x\gets D]-\Pr[x\gets D']|$.

\paragraph{Quantum circuits.} A \emph{qubit} is the special case where $\Hs$ has dimension 2, often denoted $\Hs_2$. We will typically consider Hilbert spaces that are the product of many qubits: $\Hs=\Hs_2^{\otimes n}$.

We now describe our non-uniform model of computation using quantum circuits, following~\cite{STOC:AhaKitNis98}. Let $\Gs$ be a fixed, finite set of operations. We will assume each operation in $\Gs$ is length preserving, meaning it makes $k$ qubits to $k$ qubits for some $k$ (different operations in $\Gs$ may have different $k$). We will call the elements of $\Gs$ \emph{gates}. We will always assume $\Gs$ is closed under Hermitian transpose. A unitary circuit is composed of a sequence of applications of unitary gates, and a general quantum circuit is composed of a sequence of applications of general quantum gates. The qubits are then partitioned into three sets: $\Hs_{\sf in}$, which contains the input state, $\Hs_{\sf out}$, which will contain the output state, and $\Hs_{\sf work}$, which will contain private work space. The classical input $x$ is loaded into the register $\Hs_{\sf in}$, denoted as $|x\rangle$, and then $\Hs_{\sf out}$ and $\Hs_{\sf work}$ are initialized to fixed states, which will both be denoted $|0\rangle$. At the end of the computation, $\Hs_{\sf in}$ and $\Hs_{\sf work}$ are traced out and $\Hs_{\sf out}$ is measured to get the final output.

For a quantum circuit $C$ and a classical input $x$, we will let $C(x)$ denote the distribution of outputs obtained by computing $C|x\rangle$ and then measuring $\Hs_{\sf out}$.

In general, we will consider the gate set as being a property of a quantum circuit $C$, which we will denote as $\Gs(C)$. We note that we allow $C$ to not use all the gates in the gate set, meaning $\Gs(C)$ may include gates not used in $C$.

\paragraph{Complexity Metrics.} The \emph{time complexity} of a quantum circuit is the number of gates in the circuit. The \emph{space complexity} is the sum of the number of qubits in $\Hs_{\sf in},\Hs_{\sf out},\Hs_{\sf work}$. 
\begin{remark}The above space metric is not well-suited to the regime of space sub-linear in the input size. But this can be easily handled in a number of ways, such as by having $\Hs_{\sf in}$ only being used as control qubits, and not counting it in the space. These details will not be important for us.
\end{remark}
\begin{remark}It is also possible to consider uniform quantum computational models~\cite{Wat04,MelWat12,STOC:Ta-Shma13}. However, as we discuss in Section~\ref{sec:compilers}, working in a non-uniform model makes our result stronger, as we are interested in lower-bounds.
\end{remark}

\paragraph{Oracle-assisted Circuits.} An oracle-assisted circuit is one that may make queries to a unitary $\Um$. For oracle-assisted circuits, the time complexity is the total of the number of gates and number of oracle calls to $\Um$. The space complexity is still the number of qubits in $\Hs_{\sf in},\Hs_{\sf out},\Hs_{\sf work}$. These time and space complexities do \emph{not} include the time and space used internally by $\Um$.

\paragraph{Universality.} A \emph{universal unitary gate set} is a finite set of unitary gates $\Gs$, such that any unitary operation can be approximately arbitrarily closely by circuits over $\Gs$. A \emph{universal general gate set} is a finite set of general gates $\Gs$, such that any CPTP map can be approximated arbitrarily closely by circuits over $\Gs$. A \emph{universal measurement set} is a set of gates $\Ms$, such that for any universal unitary gate set $\Gs$, $\Gs\cup \Ms$ is a universal general gate set. An example of a universal measurement set is $\Ms=\{{\sf Trash},{\sf Init}\}$, where ${\sf Trash}$ traces out a qubit (outputting nothing) and ${\sf Init}$ initializes a new qubit to a fixed state typically denoted as $|0\rangle$. Note that any universal general gate set is also a universal measurement set. A \emph{proper} universal general gate set has the form $\Gs\cup\Ms$ where $\Gs$ is a universal unitary gate set and $\Ms$ is a universal measurement set.

Note that the unital measurement gate from~\cite{ITCS:GirRaz22} is \emph{not} universal, since when combined with unitary gates it only gives unital circuits.

\paragraph{Symmetric Subspaces} For a Hilbert space $\Hs$ and positive integer $\ell$, let $\Sym^\ell\Hs$ be the symmetric subspace of $\ell$ copies of $\Hs$, which is the space of all states that are invariant under permuting the $\ell$ copies of $\Hs$. The symmetric subspace has dimension ${\sf Dim}\left(\Sym^\ell\Hs\right)=\binom{{\sf Dim}(\Hs)+\ell-1}{\ell}$. We will somewhat abuse notation, and also let $\Sym^\ell$ denote the projection of the space $\Hs^{\otimes \ell}$ onto the symmetric subspace $\Sym^\ell\Hs$.

\paragraph{Haar Random States.} We will avoid specifying the formal definition of Haar random states, but will make use of a few key facts. First is that the density matrix of $\ell$ copies of a Haar random state over system $\Hs$ is identical to that of the totally mixed state over $\Sym^\ell\Hs$. Second is a no-cloning statement, which says that the optimal probability of constructing $|\psi\rangle^{\otimes(\ell+1)}$ from $|\psi\rangle^{\otimes\ell}$ is at most the ratio of the dimensions of the symmetric subspaces $\Sym^\ell\Hs$ and $\Sym^{\ell+1}\Hs$, which works out to be $\ell/({\sf Dim}(\Hs)+\ell)$~\cite{Werner98}.

\paragraph{Reflections and Projections.} For a state $|\psi\rangle$, let $P_{|\psi\rangle} := 1- 2|\psi\rangle\langle\psi|$ be the reflection about $|\psi\rangle$. We observe that $P_{|\psi\rangle}$ can be used to implement the map satisfying $|\psi\rangle|b\rangle\mapsto|\psi\rangle|b\oplus 1\rangle$ and identity on all states orthogonal to $|\psi\rangle|b\rangle$. Indeed, we can apply the Hadamard transform to $|b\rangle$, obtaining $\frac{1}{\sqrt{2}}|0\rangle+(-1)^b\frac{1}{\sqrt{2}}|1\rangle$. Then controlled on the bit in this register, we apply $P_{|\psi\rangle}$. In the case where the state is $|\psi\rangle$, then $|0\rangle$ maps to $|0\rangle$ while $|1\rangle$ maps to $-|1\rangle$. This maps the overall qubit state to $\frac{1}{\sqrt{2}}|0\rangle+(-1)^{b+1}\frac{1}{\sqrt{2}}|1\rangle$; applying Hadamard one more time give $|b\oplus 1\rangle$. On the other hand, if the state is $|\tau\rangle$ orthogonal to $|\psi\rangle$, then $P_{|\psi\rangle}$ acts as the identity.
		
Using the latter formulation, we can also implement the projective measurement $|\psi\rangle\langle\psi|$ by simply initializing the qubit to 0, applying the transformation above, and then measuring the qubit. If we get a 1, we know the state is $|\psi\rangle$, while a zero tells us that the state is orthogonal to $|\psi\rangle$. We will abuse notation, and let $P_{|\psi\rangle}$ whichever version (reflection, $|\psi\rangle|b\rangle\mapsto|\psi\rangle|b\oplus 1\rangle$, or projection) is most convenient. 

\paragraph{Queries to Classical Functions.} Given a classical function $O:\{0,1\}^m\rightarrow\{0,1\}^n$, we can have an algorithm make queries to $O$. Do do so, we turn $O$ into a unitary $\Om$ that acts on $\Hs_2^{\otimes m}\otimes\Hs_2^{\otimes n}$ as $\Om|x,y\rangle=|x,y\oplus O(x)\rangle$. Then any query to $O$ simply applies the unitary $\Om$.

\paragraph{Some Useful Lemmas.} Consider the state $|\phi_t\rangle=\sum\alpha_{x,y}|x,y\rangle$ of a quantum query algorithm when it makes its $t$-th quantum query. Define $q_x(|\phi_t\rangle)$ to be the magnitude squared of $x$ in the superposition of query $t$, that is $q_x(|\phi_t\rangle)=\sum_y |\alpha_{x,y}|^2$. Call this the query magnitude of $x$. Let $q_x=\sum_t q_x(|\phi_t\rangle)$ be the total query magnitude of $x$. For a set $S$, let $q_S=\sum_{x\in S}q_x$ be the total query magnitude of $S$.

\begin{lemma}[\cite{BBBV97} Theorem 3.1]\label{lem:bbbv1} Suppose $\||\phi\rangle-|\psi\rangle\|\leq\epsilon$. Then performing any measurement measurement on $|\phi\rangle$ and $|\psi\rangle$ yields distributions with statistical distance at most $4\epsilon$.\end{lemma}

\begin{lemma}[\cite{BBBV97} Theorem 3.3]\label{lem:bbbv2} Let $\As$ be a quantum query algorithm making $T$ queries to an oracle $O$. Let $\epsilon > 0$ and let $S$ be a set such that $q_S \leq \epsilon$. Let $O'$ be another oracle that is identical to $O$ on all points not in $S$. Let $|\phi\rangle,|\psi\rangle$ be the final state of $\As$ when given $O,O'$, respectively. Then $\||\phi\rangle-|\psi\rangle\|\leq\sqrt{T\epsilon}$
\end{lemma}

\section{Quantum Circuit Compilers and Black Box Purifiers}\label{sec:compilers}

Here, we give our notion of black box impossibility for circuit compilers. A \emph{property} of a quantum circuit is a function $\Ps(C)\in\{0,1\}$. We say that $C$ has property $\Ps$ if $\Ps(C)=1$. Equivalently, a property is a subset of all possible quantum circuits. Example properties include:
\begin{itemize}
	\item The ``all circuits'' property $P_{\mathbb{1}}$ defined as $P_{\mathbb{1}}(C)=1$ for all $C$.
	\item The unitary property $\Ps_{\sf Unitary}$, where $\Ps_{\sf Unitary}(C)=1$ if and only if $C$ only makes use of unitary gates\footnote{Note that $C$ may compute a unitary operation even if it contains non-unitary gates. In such a case we would say that $\Ps_{\sf Unitary}(C)=0$ despite $C$ being a unitary operation.}.
	\item The size property $\Ps_{{\sf Size}(S)}$ property, where $\Ps_{{\sf Size}(S)}(C)=1$ if and only if $C$ has size at most $S$. Likewise we can define the time property $\Ps_{{\sf Time}(T)}$.
	\item Fix a proper universal general gate set $\Gs_0$\footnote{Recall that being proper means that $\Gs_0$ can be divided into a universal unitary gate set and a universal measurement set.}. The ``normal form'' property $\Ps_{\sf Normal}$ (with respect to $\Gs_0$) is the property that (1) $\Gs_0\subseteq \Gs(C)$, and (2) that $\Gs(C)\setminus\Gs_0$ contains only unitary gates. In other words, a normal form circuit is a circuit whose non-unitary gates must come from $\Gs_0$, but the unitary gates could be arbitrary.
	\item Any combination of the above properties, such as being unitary and time $T$, which would be denoted $\Ps_{\sf Unitary}\cap\Ps_{{\sf Time}(T)}$.
\end{itemize}

\begin{definition}Let $\Ps,\Qs$ be two properties of quantum circuits. A $\Ps\Rightarrow\Qs$ compiler is a function $\Cs$ from circuits to circuits such that:
	\begin{itemize}
		\item {\bf Same gate sets}: For any quantum circuit $C$, $\Gs(\Cs(C))=\Gs(C)$\enspace\footnote{Recall that $C,\Cs(C)$ do not need to use all gates in their gate set.}.
		\item {\bf Close functionalities:} For any $C$ and any classical string $x$, $\Delta(\;C(x)\;,\;\Cs(C)(x)\;)\leq 1/3$.
		\item {\bf Property transforming:} For any $C$, if $\Ps(C)=1$, then $\Qs(\Cs(C))=1$.
	\end{itemize}
\end{definition}
	
The choice of the constant $1/3$ is arbitrary, and typically only affects an overall constant factor in the complexity of $\Cs(C)$, which would be absorbed into Big-Oh notation. Note that any compiler, by definition, maps normal-form circuits to normal-form circuits, since the input and output circuits have the same gate set. 

\begin{remark}Typically, one would want a circuit transformation to approximately preserve the action of $C$ on any \emph{quantum} input. However, as we are interested in low-bounds here, only asking for approximately preservation on classical inputs will make our results stronger.
\end{remark}

\paragraph{Purifiers.} With our notion of compilers in hand, we are now ready to give our notion of a purifier. A purifier transforms any quantum circuit with measurement gates (or more generally non-unitary gates) into a circuit with only unitary gates. A \emph{black-box} purifier, in some sense, successfully removes non-unitary gates, no matter what gate set the original circuit used. More precisely:

\begin{definition} Fix a universal general gate set $\Gs_0$. A \emph{black-box purifier} is a $\Ps_{\sf Normal}\Rightarrow\Ps_{\sf Unitary}$ compiler. We can also consider purifiers that maintain bounds on the time and/or space:
	\begin{itemize}
		\item For a function $S':\Z^+\rightarrow\Z^+$, a \emph{black-box $S'$-space purifier} is a $\Ps_{\sf Normal}\cap\Ps_{{\sf Space}(S)}\Rightarrow\Ps_{\sf Unitary}\cap\Ps_{{\sf Space}(S'(S))}$ compiler, for any $S>0$.
		\item For a function $T':\Z^+\rightarrow\Z^+$, a \emph{black-box $T'$-time purifier} is a $\Ps_{\sf Normal}\cap\Ps_{{\sf Time}(T)}\Rightarrow\Ps_{\sf Unitary}\cap\Ps_{{\sf Time}(T'(T))}$ compiler, for any $T>0$.
		\item For functions $S',T':(\Z^+)^2\rightarrow \Z^+$, a \emph{black-box $(S',T')$-space-time purifier} is a $\Ps_{\sf Normal}\cap\Ps_{{\sf Space}(S)}\cap\Ps_{{\sf Time}(T)}\Rightarrow\Ps_{\sf Unitary}\cap\Ps_{{\sf Space}(S'(S,T))}\cap\Ps_{{\sf Time}(T'(S,T))}$ compiler.
	\end{itemize}
\end{definition}

In other words, a black box purifier maps any quantum circuit that is in normal form into one that is unitary (and also in normal form, since compilers preserve gate sets). A black box purifier, in other words, removes the non-unitary gates, and must do so using the unitary part of the original gate set, \emph{regardless} of the choice of unitaries. We note, however, that our definition allows the purifier to depend arbitrarily on the gate set and circuit inputs. The only requirement is that it must work no matter the choice of unitaries in the original gate set. An $S'$ space, $T'$ time or $(S',T')$ space-time purifier must do this while outputting circuits with space $S'$ and/or time $T'$.

\paragraph{Purifies from the literature.} We now explain our notion of black box purifier captures existing approaches for removing intermediate measurements from quantum computation:
\begin{itemize}
	\item \cite{ITCS:GirRaz22}: This work takes the original quantum circuit $C$, and generates a new quantum circuit $C'$ as follows: it takes every measurement gate, and replaces it essentially with a random phase gate. The randomness for the phase gate is then derived by a suitable explicit pseudorandom generator (PRG). As a consequence, the ``unitary part'' of $C$ is entirely un-touched, and we only need to add the PRG computation which can be expressed in any universal gate set. As long as the gate set of $C$ is proper, this PRG computation can be expressed in terms of the unitary gates from the gate set. Thus, their result is black-box. 
	\item \cite{STOC:FefRem21,ICALP:GirRazZha21}: These works follow an approach where the computation is broken into a sequence of arbitrary channels $\Phi_1,\cdots\Phi_T$\enspace\footnote{In~\cite{STOC:FefRem21}, the channels are truly arbitrary. In~\cite{ICALP:GirRazZha21}, the channels are restricted to being unital. This restriction does not affect the discussion here.}. These channels are then expressed as matrices representing the transformations on the underlying Hilbert spaces, and then multiplied using a low-space multiplication algorithm to get the final result. Arbitrary channels can implement any arbitrary gate set, and the algorithm for low-space matrix multiplication can be implemented in any universal gate set. Thus we see that their results are also black-box. Note here that these results appear to be ``non-black-box'' in the sense that they make explicit use of the matrix representation of the gates. However, they are still ``black-box'' in our sense, as we allow the circuit $C'$ to depend arbitrarily on $C$ and it's gates, as long as the transformation is possible \emph{for any} starting $C$ with arbitrary gates.
	\item \cite{Bennett89}: this work is not a black-box purifier in our sense simply because it is a transformation on \emph{classical} circuits, and the goal is not to remove measurements but to make the circuit reversible. However, we can define an analogous notion of classical \emph{reversible-izers} that takes any classical circuit $C$ comprising both reversible and irreversible gates, and outputs a new circuit $C'$ comprising only of reversible gates. A reversible-izer would then be black-box as long as is worked for any set of starting gates.~\cite{Bennett89} would then be a black-box reversible-izer. Because of the similarities of purifying and reversible-izing, we would therefore expect any attempt to adapt~\cite{Bennett89} in order to purify quantum computations would result in a black-box purifier as well.
\end{itemize}
\begin{remark}\cite{STOC:FefRem21,ICALP:GirRazZha21,ITCS:GirRaz22} work in the uniform setting where the quantum circuits are generated uniformly by a classical Turing machine, and additional space and time constraints are placed on the Turing machine. Our notion of a purifier is more lax, as it does not place any resource constraints on how the circuits are generated. In our lower bound in Sections~\ref{sec:sep},~\ref{sec:main}, the starting low-space algorithm that contains measurements is easily seen to be uniformly generated. Since our lower bound will apply even to compilers producing non-uniform quantum circuits, the laxness of our purifier notion makes our results even stronger.
\end{remark}

\section{A separation between pure and general quantum computation}\label{sec:sep}

Before proving our main theorem, we first prove a slightly weaker theorem separating pure and general computation that captures the main technical challenges in our main theorem. 

In short, here we show that we can construct a gate relative to which there is a bit $b$ that can be computed by a low space-time algorithm with measurements, but no low space-time \emph{unitary} algorithm can compute $b$. Here, it is crucial that the algorithms are independent of the choice of gate. If the algorithm is allowed to depend on the choice of gate, then the algorithm can simply have $b$ hardcoded and output that $b$. As such, our result in this section does not result in a separation for circuits that are allowed to depend on the choice of gate. Note that our notion of a compiler is allowed to depend on the gates being used, so the result from this section is insufficient. We extend this result to gate-dependent algorithms, and therefore rule out time- and space-efficient black box purifiers, in Section~\ref{sec:main}.

Consider $n,t\in\Z^+$. Assume for simplicity that $t+1=2^m$ for an integer $m$. Let $\Psi=\{|\psi_i\rangle,|\phi_i\rangle\}_{i\in[t]}$ be a list of $2t$ pure quantum states over $\Hs_2^{\otimes n}$ that are orthogonal to $|0\rangle$; denote this space as $\Hs_2^{\otimes n}\setminus\{|0\rangle\}$. We will think of these states as each being Haar random over $\Hs_2^{\otimes n}\setminus\{|0\rangle\}$. Let $o\in\{0,1\}$ be a bit, which we will think of as being a uniform random bit. Define the following unitary function $\Om_{\Psi,\out}$ that acts on $\Hs_2^{\otimes(m+n+n)}$:
\[\arraycolsep=0pt\begin{array}{lllllll}
	\Om_{\Psi,\out}|0\rangle&|0^n\rangle&|0^n\rangle&=|0\rangle&|\psi_1\rangle&|\phi_1\rangle\\\vspace{5pt} \Om_{\Psi,\out}|0\rangle&|\psi_1\rangle&|\phi_1\rangle&=|0\rangle&|0^n\rangle&|0^n\rangle\\
	\Om_{\Psi,\out}|i\rangle&|\psi_i\rangle&|0^n\rangle&=|i\rangle&|\psi_{i+1}\rangle&|\phi_{i+1}\rangle&\text{ for }i\in[1,t-1]\\\vspace{5pt}
	\Om_{\Psi,\out}|i\rangle&|\psi_{i+1}\rangle&|\phi_{i+1}\rangle&=|i\rangle&|\psi_i\rangle&|0^n\rangle&\text{ for }i\in[1,t-1]\\
	\Om_{\Psi,\out}|t\rangle&|\psi_t\rangle&|z\rangle&=|t\rangle&|\psi_t\rangle&|z\oplus \out 0^{n-1}\rangle&\text{ for }z\in\{0,1\}^n
\end{array}\]
Meanwhile, $\Om_{\Psi,\out}$ preserves all states orthogonal to the states above. The goal will be, given oracle access to $\Om_{\Psi,\out}$, to compute $\out$.

An alternative view of the transformation is as:
\begin{align*}
	\Om_{\Psi,\out}=|0\rangle\langle 0|\otimes &\Big[\;\Id-\Big(\;|0^n\rangle|0^n\rangle-|\psi_1\rangle|\phi_1\rangle\;\Big)\Big(\;\langle 0^n|\langle 0^n|-\langle\psi_1|\langle\phi_1|\;\Big)\;\Big]\\
	+\sum_{i=1}^{t-1}|i\rangle\langle i|\otimes &\Big[\;\Id-\Big(\;|\psi_i\rangle|0^n\rangle-|\psi_{i+1}\rangle|\phi_{i+1}\rangle\;\Big)\Big(\;\langle \psi_i|\langle 0^n|-\langle\psi_{i+1}|\langle\phi_{i+1}|\;\Big)\;\Big]\\
	+|t\rangle\langle t|\otimes &\Big[\;\Big(\Id-|\psi_t\rangle\langle\psi_t|\Big)\otimes\Id+|\psi_t\rangle\langle\psi_t|\otimes \;\Big(\sum_z |z\oplus \out 0^{n-1}\rangle\langle z|\;\Big)\;\Big]
\end{align*}

\paragraph{A Low-Space-Time Circuit with Measurements.}

\begin{lemma}\label{lem:generalcircuitupper} There exists a general quantum circuit $C$ with space $O(n+m)=O(n+\log t)$ and time $O(nt)$ that computes $\out$ probability 1.
\end{lemma}
\begin{proof}Our general quantum circuit does the following:
	\begin{itemize}
		\item Initialize registers $|0\rangle|0^n\rangle|0^n\rangle$.
		\item Apply an $\Om$ gate, to obtain $|0\rangle|\psi_1\rangle|\phi_1\rangle$
		\item Repeat the following loop for $i=1,\dots,t$, where the state at the beginning of the loop is $|i-1\rangle|\psi_i\rangle|\phi_i\rangle$:
		\begin{itemize}
			\item Reset the state $|\phi_i\rangle$ to $|0^n\rangle$
			\item Add 1 (mod $t+1$) to the first register, which contains $i-1$. At this point, the state is $|i\rangle|\psi_i\rangle|0^n\rangle$
			\item Make a query to $\Om$. If $i<t$, the resulting state is now $|i\rangle|\psi_{i+1}\rangle|\phi_{i+1}\rangle$. If $i=t$, the state is now $|t\rangle|\psi_t\rangle|\out 0^{n-1}\rangle$.
		\end{itemize}
		\item Discard $|t\rangle,|\psi_t\rangle$, and $|0^{n-1}\rangle$, leaving $|\out\rangle$. Measure and output $\out$.
	\end{itemize}
Above, resetting $|\phi_i\rangle$ to $|0^n\rangle$ can be accomplished with $n$ qubit reset gates. As mentioned in Section~\ref{sec:intro}, qubit reset gates are space- and time-equivalent to many typical notions of measurement gates, assuming the ability to use the classical results of measurement to control later gates. However, we cannot reset $|\phi_i\rangle$ with a unital measurement: such a unital measurement will result in $|z\rangle$ for a random $z$, but then there is no way to overwrite $|z\rangle$ with $|0^n\rangle$. Please see Figure~\ref{fig:constrmain} for a depiction of our algorithm.

\begin{figure}[htb]
	\centering
	\includegraphics[width=.8\textwidth]{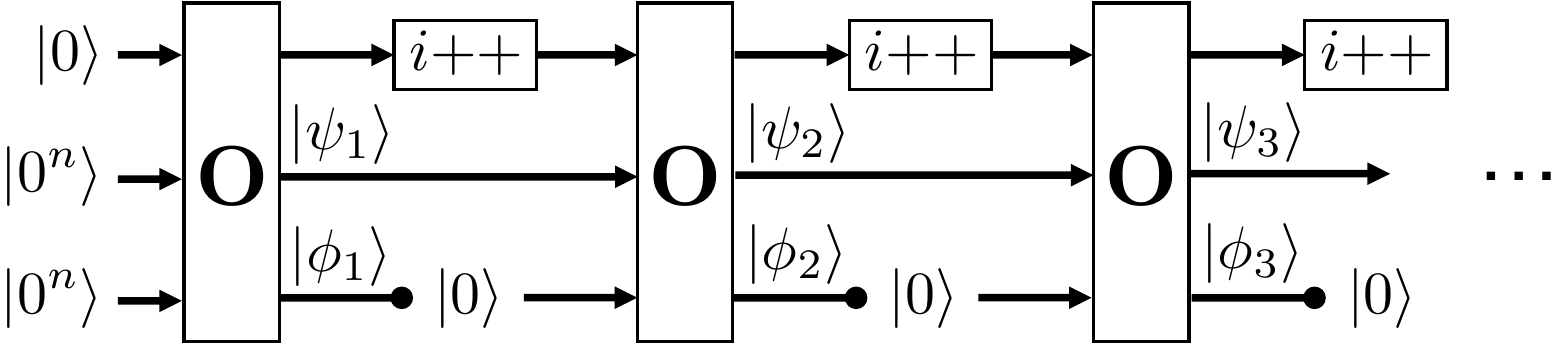}
	\captionsetup{width=.8\textwidth}
	\caption{\label{fig:constrmain} Our task that can be computed in low space and time with measurements, but requires large space or large time without.}
\end{figure}

The space of the algorithm above is $2n+m=O(n+\log(t))$, plus any extra space needed to discard and initialize new registers, which is constant. Thus, the overall space is $O(n+\log(t))$. For time, there are $t+1$ applications of the $\Om$ gate, plus in each of the $t$ iterations we have $n$ qubits are discarded and re-initialized, taking time $O(n)$ per iteration. This gives an overall number of gates equal to $O(nt)$.\end{proof}

\paragraph{No Low-Space-Time Circuit without Measurements.}

\begin{lemma}\label{lem:unitarycircuitlower} There exists a distribution over $\Psi,\out$ and constants $c',d'$ such that, for any unitary circuit over any gate set which makes at most $2^{d'n}$ queries to $\Om_{\Psi,\out}$ and runs in space $S\leq c'nt$, the probability of outputting $\out$ is less than $7/12$.\end{lemma}
The constant $7/12$ above is arbitrary, as long as it is strictly between $1/2$ and $2/3$ (the latter being the arbitrary constant in the definition of a black box purifier). Above, note that the time $T$ of the circuit is at least the number of queries to $\Om_{\Psi,\out}$. As a consequence, there is no time $T\leq 2^{d' n}$ and space $S\leq c'nt$ circuit that can guess $\out$ with probability at least $7/12$. This holds independent of the gate set that the algorithm uses, as long as the gates do not depend of the choice of $\Psi,\out$. This independence from the gate set will be important for our ultimate result in Section~\ref{sec:main}. There, we will construct from any low space-time purifier a circuit which contradicts Lemma~\ref{lem:unitarycircuitlower}. Our circuit, however, will have to make some computations that are potentially very expensive, but fortunately are independent of $\Psi,\out$. In order for these extra computations to still contradict Lemma~\ref{lem:unitarycircuitlower}, we state the lemma as only counting queries to $\Om_{\Psi,\out}$ but being otherwise independent of the actual time-complexity of the circuit.

\noindent The rest of this section will be devoted to proving Lemma~\ref{lem:unitarycircuitlower}. 

\paragraph{Roadmap.} We will assume an algorithm $\As$ running in time much less than $2^{O(n)}$ with probability of outputting $\out$ being at least $7/12$. We will show that such an algorithm must have space $\Omega(nt)$. First, we will show that any unitary algorithm that outputs $\out$ with significant probability must actually be able to produce $|\psi_t\rangle$. This follows from standard quantum query techniques. Then, we will design a simulator which approximately simulates $\Psi$ using only several copies of the $|\psi_i\rangle,|\phi_i\rangle$ instead of the full descriptions of these states. This simulation uses ideas from~\cite{C:JiLiuSon18}, and will cause some error which will be small assuming the unitary algorithm's running time is small. As the unitary algorithm is run, some of these copies will be provided to the algorithm, decreasing the storage of the simulator. We show essentially that the simulator must have given at least one copy of $|\phi_i\rangle$ for each $i$ to the algorithm in order for the algorithm to have obtained $|\psi_t\rangle$. This implies an upper bound on the space of the simulator at the end of the computation. Finally, we observe that the total joint storage of the simulator and the unitary algorithm cannot drop below the initial simulator storage. This then implies a lower bound on the space of the unitary algorithm.

\subsection{From Computing $\out$ to computing $|\psi_t\rangle$}

\begin{lemma}Fix any $\Psi$. Let $\As$ be a unitary algorithm making queries to $\Om_{\Psi,\out}$ running in time $T$ and space $S$, such that $\Pr_{o\gets\{0,1\}}[\As^{\Om_{\Psi,\out}}()=o]\geq 7/12$. Then there is a another unitary algorithm $\As_1$ running in time $\leq T$ and space $\leq S$ that attempts to output $|\psi_t\rangle$ with the following guarantee. If $\rho$ is the final state of $\As_1$ when making queries to $\Om_{\Psi,\out}$ for a random $\out$, then $\langle\psi_t|\rho|\psi_t\rangle\geq (48T)^{-2}$.
\end{lemma}
\begin{proof}Since $\As$ runs in time at most $T$, it makes at most $T$ queries to $\Om$. We zoom in on the basis states of the queries to $\Om$ where the first register is $|t\rangle$. Pick any basis for $\Hs_2^{\otimes n}$ which contains $|\psi_t\rangle$ as the first element, and look at the queries to $\Om$ in that basis. Then we see that $\Om$, when restricted to the first register being $|t\rangle$, is implementing a quantum query to a classical function, namely the function that maps $0$ (corresponding to the first basis element being $|\psi_t\rangle$) to $\out$, and everything else to 0. Let the total query weight on the basis element $|\psi_t\rangle$ be $\epsilon$. We now switch from $\Om_{\Psi,\out}$ to $\Om_{\Psi,0}$, which contains no information about $\out$. This means $\As$ when querying $\Om_{\Psi,0}$ outputs $\out$ with probability $1/2$. By Lemma~\ref{lem:bbbv2}, this change moves the output state of $\As$ by at most $\sqrt{T\epsilon}$. Then applying Lemma~\ref{lem:bbbv1} shows that the output distribution is affected by at most $4\sqrt{T\epsilon}$. In other words, $|p-1/2|\leq 4\sqrt{T\epsilon}$. By our assumption that $p\geq 7/12$, we have that $\epsilon\geq (48^2T)^{-1}$. 
	
This means there is some query $j\in[T]$ such that the query weight on $|\psi_t\rangle$ is at least $\epsilon/T\geq(48T)^{-2}$. We therefore define $\As_1$ as the algorithm which runs $\As$ until query $j$, and outputs the middle register of the query. See Figure~\ref{fig:computephit}.
\begin{figure}[htb]
	\centering
	\includegraphics[width=.8\textwidth]{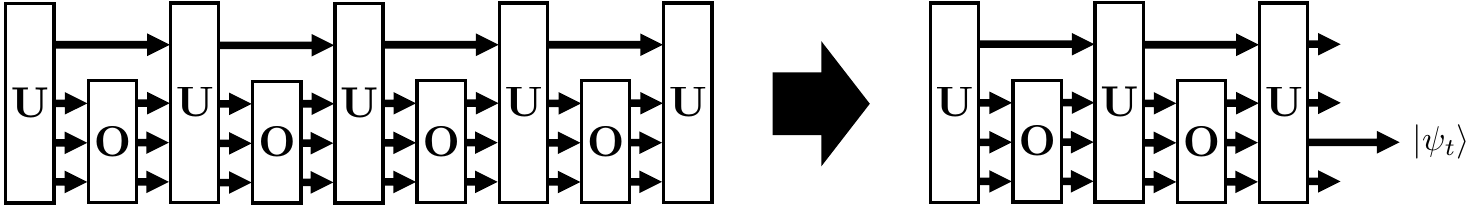}
	\captionsetup{width=.8\textwidth}
	\caption{\label{fig:computephit} Converting from an algorithm $\As$ (top) which computes $\out$ to an algorithm $\As_1$ (bottom) which computes $|\psi_t\rangle$.}
\end{figure}
By the analysis above, if $\rho$ is the output state of $\As_1$, then $\langle\psi_t|\rho |\psi_t\rangle\geq (48T)^{-2}$. Observe that the number of queries $\As_1$ makes and the space of $\As_1$ are at most the query count and space for $\As$. If $\As$ is unitary, then so is $\As_1$.\end{proof}

\noindent From now on, we will assume an algorithm $\As$ which outputs $|\psi_t\rangle$.

\subsection{Simulating $\Om_{\Psi,\out}$: Counting}\label{sec:counting}

We now gradually build up a simulator which simulates $\Om_{\Psi,\out}$ to the algorithm. Our ultimate simulator will only use several copies of each of the $|\psi_i\rangle$ and $|\phi_i\rangle$.

As a first step, we show that we can simulate $\Om_{\Psi,\out}$ as specified above, but we can record in some ancilla registers the net number of copies of each of the $|\psi_i\rangle,|\phi_i\rangle$ that have given out/consumed\footnote{By ``net'', we mean the difference between the number given out minus the number consumed.}.

Let $\Hs_\infty$ be the infinite-dimensional Hilbert space spanned by $\{|x\rangle\}_{x\in\Z}$. Note that we use an infinite dimensional space for simplicity, and we can make the space finite-dimensional by instead using the Hilbert space spanned by $\{-U,-U+1,\cdots,U-1,U\}$ for some $U\geq T$.

For $i=1,\dots,t$ let $\Cs_i$ be a copy of $\Hs_\infty$ and likewise for $i=1,\dots,t-1$ let $\Ds_i$ be a copy of $\Hs_\infty$. Each of the $\Cs_i,\Ds_i$ will be initialized to $|0\rangle$. Let $\Cs$ be the joint system of all $\Cs_i,\Ds_i$. We will write the basis states of $\Cs$ as $|\cv\rangle$ where $\cv\in\Z^{2t-1}$.

Let ${\sf Incr}_{\Cs_i}$ and ${\sf Decr}_{\Cs_i}$ be the operation on $\Cs$ which applies the map $|j\rangle\mapsto|j+1\rangle$ and $|j\rangle\mapsto|j-1\rangle$, respectively, to register $\Cs_i$. Likewise define ${\sf Incr}_{\Ds_i}$ and ${\sf Decr}_{\Ds_i}$.

Now define the following unitary $\Om^{(1)}_{\Psi,\out}$ that acts on $\Hs_2^{\otimes(m+n+n)}\otimes \Cs$:
\[\arraycolsep=0pt\begin{array}{lllllllll}
	\Om^{(1)}_{\Psi,\out}|0\rangle&|0^n\rangle&|0^n\rangle&|\cv\rangle&=|0\rangle&|\psi_1\rangle&|\phi_1\rangle&{\sf Incr}_{\Cs_1}{\sf Incr}_{\Ds_1}|\cv\rangle\\\vspace{5pt} \Om^{(1)}_{\Psi,\out}|0\rangle&|\psi_1\rangle&|\phi_1\rangle&|\cv\rangle&=|0\rangle&|0^n\rangle&|0^n\rangle&{\sf Decr}_{\Cs_1}{\sf Decr}_{\Ds_1}|\cv\rangle\\
	\Om^{(1)}_{\Psi,\out}|i\rangle&|\psi_i\rangle&|0^n\rangle&|\cv\rangle&=|i\rangle&|\psi_{i+1}\rangle&|\phi_{i+1}\rangle&{\sf Decr}_{\Cs_{i}}{\sf Incr}_{\Cs_{i+1}}{\sf Incr}_{\Ds_{i+1}}|\cv\rangle&\text{ for }i\in[1,t-1]\\\vspace{5pt}
	\Om^{(1)}_{\Psi,\out}|i\rangle&|\psi_{i+1}\rangle&|\phi_{i+1}\rangle&|\cv\rangle&=|i\rangle&|\psi_i\rangle&|0^n\rangle&{\sf Incr}_{\Cs_{i}}{\sf Decr}_{\Cs_{i+1}}{\sf Decr}_{\Ds_{i+1}}|\cv\rangle&\text{ for }i\in[1,t-1]\\
	\Om^{(1)}_{\Psi,\out}|t\rangle&|\psi_t\rangle&|z\rangle&|\cv\rangle&=|t\rangle&|\psi_t\rangle&|z\oplus \out 0^{n-1}\rangle&|\cv\rangle&\text{ for }z\in\{0,1\}^n
\end{array}\]
In other words, any time $\Om^{(1)}_{\Psi,\out}$ outputs a $|\psi_i\rangle$, it increments the corresponding $\Cs_i$ register; and analogously increments $\Ds_i$ whenever it outputs a $|\phi_i\rangle$. On the other hand, any time $\Om^{(1)}_{\Psi,\out}$ must absorb a $|\psi_i\rangle$ or $|\phi_i\rangle$, it decrements the analogous register.

Above, we will think of $\Cs$ as being in the private state of a simulator, inaccessible to the algorithm. We now demonstrate that $\Om^{(1)}_{\Psi,\out}$ is actually indistinguishable from $\Om_{\Psi,\out}$, for certain distributions over $\Psi$. Note that while $\Om^{(1)}_{\Psi,\out}$ and $\Om_{\Psi,\out}$ may look like they act identically on $\Hs_2^{\otimes(m+n+n)}$, the fact that $\Om^{(1)}_{\Psi,\out}$ is modifying external registers based on the contents of $\Om^{(1)}_{\Psi,\out}$ means that the operations are, in fact, not identical on $\Hs_2^{\otimes(m+n+n)}$. Essentially, by adding/subtracting from $\Cs$, $\Om^{(1)}_{\Psi,\out}$ may split different branches of the computation, eliminating interference that may be present with $\Om_{\Psi,\out}$. In particular, if $\Psi$ is fixed and known to the algorithm, it is not hard to design an algorithm which can successfully distinguish between $\Om^{(1)}_{\Psi,\out}$ and $\Om_{\Psi,\out}$. We show, however, that if $\Psi$ is chosen from a ``sufficiently random'' collection of states, then such distinguishing is not possible.

\begin{definition}We say a distribution $\Ds$ over states in $\Hs_2^{\otimes n}$ is \emph{phase invariant} if $\Ds$ is identical to a distribution of the following form:
	\begin{itemize}
		\item Sample a state $|\psi'\rangle$ from some distribution $\Ds'$.
		\item Choose for each $x\in\{0,1\}^n$ a uniform real number $\tau_x\in(-\pi,\pi]$
		\item Apply to $|\psi'\rangle$ the operation which maps $|x\rangle\mapsto e^{i\tau_x}|x\rangle$ for each $x$.
		\item Output the resulting state $|\psi\rangle$.
	\end{itemize}
We say $\Ds$ is $M$-phase invariant if $\tau_x$ is instead uniform on the multiples of $2\pi/M$ in $(-\pi,\pi]$.
\end{definition}
We note that any phase invariant distribution is also $M$-phase invariant, since we can absorb the ``extra'' randomness of $\tau_x$ into the distribution of $|\psi'\rangle$. We also observe that Haar random states are phase invariant.

\begin{lemma}\label{lem:count} Let $\As$ be a time $T$ algorithm. Then for any $M\geq 2T+1$ and any $M$-phase invariant distribution $\Ds$, if $\Psi\gets\Ds^{2t-1}$ (meaning the $|\psi_i\rangle$ and $|\phi_i\rangle$ are sampled from $\Ds$), then
	\[\E_\Psi\left[\As^{\Om_{\Psi,\out}}()\right]=\E_\Psi\left[\As^{\Om^{(1)}_{\Psi,\out}}()\right]\]
\end{lemma}
\begin{proof}Let $\As,M,\Ds$ be as in the statement of Lemma~\ref{lem:count}. Let $\Z_M$ be the set of of integers mod $M$, which we will associate with the interval $[-\lfloor (M-1)/2\rfloor,\cdots,\lceil (M-1)/2\rceil]$. 
	
Consider sampling each $|\psi_i\rangle$ and $|\phi_i\rangle$ from $\Ds$. Let $|\psi_i'\rangle,|\phi_i'\rangle$ be the samples from $\Ds'$ used to sample $|\psi_i\rangle,|\phi_i\rangle$ according to the definition of $M$-phase invariance. For each state $|\psi_i\rangle$, let $F_i:\{0,1\}^n\rightarrow\Z_M$ be the function mapping $x$ to $\tau_x*(M/2\pi)$ where the $\tau_x$ are the random values used to construct $|\psi_i\rangle$. Likewise define $G_i:\{0,1\}^n\rightarrow\Z_M$ as the functions mapping $x$ to $\tau_x*(M/2\pi)$ for each of the $|\phi_i\rangle$. Then the $F_i,G_i$ are uniform random functions.

Now fix each $|\psi_i'\rangle,|\phi_i'\rangle$. We will show that $\Pr[\As^{\Om_{\Psi,\out}}()=z]=\Pr[\As^{\Om^{(1)}_{\Psi,\out}}()=z]$ holds even when fixing these states. Now the only randomness is over the choice of $F_i,G_i$ and any randomness of $\As$.

For a function $F:\{0,1\}^n\rightarrow\Z_M$, let $\Phase_F$ be the unitary that maps $|x\rangle\mapsto e^{i2\pi F(x)/M}|x\rangle$. Then $|\psi_i\rangle=\Phase_{F_i}|\psi_i'\rangle$ and $|\phi_i\rangle=\Phase_{G_i}|\phi_i'\rangle$
	
We will use a variant of the query recording technique of~\cite{C:Zhandry19}. Instead of sampling random $F_i,G_i$, we will purify them, initializing uniform superpositions $\frac{1}{M^{2^n/2}}\sum_{F_i}|F_i\rangle$ and $\frac{1}{M^{2^n/2}}\sum_{G_i}|G_i\rangle$. Let $\Fs$ be the register containing all the $F_i,G_i$. Now we can think of $\Om_{\Psi,\out}$ as a larger unitary $\Om_\out$ acting on $\Hs_2^{\otimes(m+n+n)}\otimes\Fs$, and $\Om^{(1)}_{\Psi,\out}$ as a larger unitary $\Om^{(1)}_\out$ acting on $\Hs_2^{\otimes(m+n+n)}\otimes\Fs\otimes\Cs$. $\Om_{\out}$ has the following behavior:

\[\arraycolsep=0pt\begin{array}{llllllll}
	\Om_{\out}|0\rangle&|0^n\rangle&|0^n\rangle&|\{F_i\}_i,\{G_i\}_i\}\rangle&=|0\rangle&\Phase_{F_1}|\psi_1'\rangle&\Phase_{G_1}|\phi_1'\rangle&|\{F_i\}_i,\{G_i\}_i\}\rangle\\\vspace{5pt} \Om_{\out}|0\rangle&\Phase_{F_1}|\psi_1'\rangle&\Phase_{G_1}|\phi_1'\rangle&|\{F_i\}_i,\{G_i\}_i\}\rangle&=|0\rangle&|0^n\rangle&|0^n\rangle&|\{F_i\}_i,\{G_i\}_i\}\rangle\\
	\Om_{\out}|i\rangle&\Phase_{F_i}|\psi_i'\rangle&|0^n\rangle&|\{F_i\}_i,\{G_i\}_i\}\rangle&=|i\rangle&\Phase_{F_{i+1}}|\psi_{i+1}'\rangle&\Phase_{G_{i+1}}|\phi_{i+1}'\rangle&|\{F_i\}_i,\{G_i\}_i\}\rangle\\\vspace{5pt}
	\Om_{\out}|i\rangle&\Phase_{F_{i+1}}|\psi_{i+1}'\rangle&\Phase_{G_{i+1}}|\phi_{i+1}'\rangle&|\{F_i\}_i,\{G_i\}_i\}\rangle&=|i\rangle&\Phase_{F_i}|\psi_i'\rangle&|0^n\rangle&|\{F_i\}_i,\{G_i\}_i\}\rangle\\
	\Om_{\out}|t\rangle&\Phase_{F_t}|\psi_t'\rangle&|z\rangle&|\{F_i\}_i,\{G_i\}_i\}\rangle&=|t\rangle&\Phase_{F_T}|\psi_t'\rangle&|z\oplus \out 0^{n-1}\rangle&|\{F_i\}_i,\{G_i\}_i\}\rangle
\end{array}\]
Meanwhile, $\Om^{(1)}_\out$ has the same behavior, except that it also acts on $\Cs$ using the ${\sf Incr},{\sf Decr}$ operations.

Now we switch to viewing the $\Fs$ register in the Fourier basis. To do so, we use the following:
\begin{lemma}\label{lem:compressed} Consider an algorithm making queries in one of two worlds. In the first, $\Fs$ is initialized to the uniform superposition over tuples of random functions, and the algorithm makes queries to $\Om_\out$ (resp. $\Om^{(1)}_\out$). In the other world, $\Fs$ is initialized to a list of all-zero functions $|(0^{2^n})^{2t-1}\rangle$, and the algorithm makes queries to $\QFT_\Fs^\dagger\cdot\Om_\out\cdot\QFT_\Fs$ (resp. $\QFT_\Fs^\dagger\cdot\Om^{(1)}_\out\cdot\QFT_\Fs$). Then the output distributions in the two worlds are equal.
\end{lemma}
\begin{proof}We prove the $\Om_\out$ case, the other being essentially identical. We insert $\QFT\cdot\QFT^\dagger=\Id$ applied to $\Fs$ between each query to $\Om_\out$. We likewise observe that the initial state of $\Fs$, the uniform superposition over all tuples of functions, is just $\QFT$ applied to $|(0^{2^n})^{2t-1}\rangle$. We can also apply a final $\QFT^\dagger$ to $\Fs$ at the very end of the computation, which does not affect the algorithm's registers. Now we observe that each $\QFT,\QFT^\dagger$ we injected commutes with the algorithm's gates other than $\Om_\out$, so we can take one half of each  $\QFT\cdot\QFT^\dagger$ and push it to being next to the previous query to $\Om_\out$ and push the other half to being next to the subsequent query. The result is each query to $\Om_\out$ is sandwiched between two $\QFT$ gates. In other words, the algorithm is now making queries to $\QFT^\dagger\cdot\Om_\out\cdot\QFT$, and the initial state of $\Fs$ is $|(0^{2^n})^{2t-1}\rangle$. These changes are all perfectly indistinguishable to the algorithm. This completes the proof of Lemma~\ref{lem:compressed}.\end{proof}

We now observe that $\QFT_\Fs^\dagger\cdot\Om_\out\cdot\QFT_\Fs$ (resp. $\QFT_\Fs^\dagger\cdot\Om^{(1)}_\out\cdot\QFT_\Fs$) have particularly nice forms. We start by observing that if we define $\Phase|y\rangle=e^{i2\pi y/M}|y\rangle$, then 
\[\QFT^\dagger\cdot \Phase\cdot\QFT |z\rangle=\frac{1}{M}\sum_{z',y}e^{-i2\pi yz'/M}e^{i2\pi y/M}e^{i2\pi yz/M}|z'\rangle=\frac{1}{M}\sum_{z',y}e^{i2\pi y(z+1-z')/M}|z'\rangle=|z+1\rangle\]
where above we used that $\sum_y e^{i2\pi yw/M}$ equals $M$ if $w=0$ and equals 0 otherwise.

Using this identity, we will interpret $\Fs$ as a collection $V$ of $2t-1$ tables, each table having length $2^n$ and containing entries from $\Z_M$. For a string $x$, let ${\sf Incr}_{\Fs_i(x)}$ be the operation that adds 1 (mod $M$) to the entry of register $\Fs$ corresponding to the $i$th function $F_i$ on input $x$. Likewise define ${\sf Decr}_{\Fs_i(x)},{\sf Incr}_{\Gs_i(x)},{\sf Decr}_{\Gs_i(x)}$. Let $|\psi_i'\rangle=\sum_x \alpha^{(i)}_x|x\rangle$ and $|\phi_i'\rangle=\sum_x \beta^{(i)}_x|x\rangle$. Then $\Pm_\out:=\QFT_\Fs^\dagger\cdot\Om_\out\cdot\QFT_\Fs$ is just
{\footnotesize
\[\arraycolsep=0pt\begin{array}{llllllll}
	\Pm_{\out}|0\rangle&|0^n\rangle&|0^n\rangle&|V\rangle&=|0\rangle&\sum\limits_{x,y}\alpha^{(1)}_x\beta^{(1)}_y |x\rangle&|y\rangle&{\sf Incr}_{\Fs_1(x)}{\sf Incr}_{\Gs_1(y)}|V\rangle\\\vspace{5pt} 
	\Pm_{\out}|0\rangle&\sum\limits_{x,y}\alpha^{(1)}_x\beta^{(1)}_y |x\rangle&|y\rangle&{\sf Incr}_{\Fs_1(x)}{\sf Incr}_{\Gs_1(y)}|V\rangle&=|0\rangle&|0^n\rangle&|0^n\rangle&|V\rangle\\
	\Pm_{\out}|i\rangle&\sum\limits_{x}\alpha^{(i)}_x |x\rangle&|0^n\rangle&{\sf Incr}_{\Fs_i(x)}|V\rangle&=|i\rangle&\sum\limits_{x',y'}\alpha^{(i+1)}_{x'}\beta^{(i+1)}_{x'} |x'\rangle&|y'\rangle&{\sf Incr}_{\Fs_{i+1}(x')}{\sf Incr}_{\Gs_{i+1}(y')}|V\rangle\\\vspace{5pt}
	\Pm_{\out}|i\rangle&\sum\limits_{x',y'}\alpha^{(i+1)}_{x'}\beta^{(i+1)}_{x'} |x'\rangle&|y'\rangle&{\sf Incr}_{\Fs_{i+1}(x')}{\sf Incr}_{\Gs_{i+1}(y')}|V\rangle&=|i\rangle&\sum\limits_{x}\alpha^{(i)}_x |x\rangle&|0^n\rangle&{\sf Incr}_{\Fs_i(x)}|V\rangle\\
	\Pm_{\out}|t\rangle&\sum\limits_{x}\alpha^{(t)}_x |x\rangle&|z\rangle&{\sf Incr}_{\Fs_t(x)}|V\rangle&=|t\rangle&\sum\limits_{x}\alpha^{(i)}_x |x\rangle&|z\oplus \out 0^{n-1}\rangle&{\sf Incr}_{\Fs_t(x)}|V\rangle
\end{array}\]}

Likewise we can define $\Pm^{(1)}_\out:=\QFT_\Fs^\dagger\cdot\Om^{(1)}_\out\cdot\QFT_\Fs$, and obtain equations for the definition of $\Pm^{(1)}_\out$, which look exactly like those for $\Pm^{(1)}_\out$, except that they include additionally the operations ${\sf Incr},{\sf Decr}$ on the register $\Cs$. 

\begin{claim}Each register $\Cs_i$ in $\Cs$ contains exactly the sum of the entries of $\Fs_i$ (when interpreted as integers in the interval $[-\lfloor (M-1)/2\rfloor,\cdots,\lceil (M-1)/2\rceil]$), and likewise $\Ds_i$ contains exactly the sum of the entries of $\Gs_i$.\end{claim}
\begin{proof}This is true initially since all the registers are zero. Then the property is preserved $\bmod M$ since, any time ${\sf Incr}_{\Fs_i(x)}$ is applied, then so is ${\sf Incr}_{\Cs_i}$, and likewise for ${\sf Incr}_{\Gs_i(x)}$ and ${\sf Incr}_{\Ds_i}$. Since the number of queries is less than $M/2$ and each query only increases or decreases the value of any register by 1, the entries in $\Fs_i,\Gs_i$ always remain in the interval $[-\lfloor (M-1)/2\rfloor,\cdots,\lceil (M-1)/2\rceil]$ and never need to be reduced $\bmod M$. Therefore, equality holds over the integers.\end{proof}

Since the register $\Cs$ can be computed from $\Fs$, which is local to the oracle simulation and not seen by the algorithm, we can imagine computing $\Cs$ from $\Fs$ immediately before each query, and then uncomputing $\Cs$ from $\Fs$ immediately after each query, and this change will not affect the algorithm in any way. The result is that we move from applying $\Pm_\out$ to $\Pm_\out^{(1)}$ without any affect on the algorithm. This shows that $\Pm_\out$ to $\Pm_\out^{(1)}$, and hence $\Om_\out$ to $\Om_\out^{(1)}$, are perfectly indistinguishable. This completes the proof of Lemma~\ref{lem:count}.
\end{proof}

We next observe the following feature of $\Om_{\Psi,\out}$:
\begin{lemma}\label{lem:countconstraints}At all times, for $i=1,\dots,t-1$, the support of $\Cs$ is on states where the count in register $\Ds_{i+1}$ is equal to the count in register $\Ds_i$ minus the count in register $\Cs_i$.
\end{lemma}
In other words, the net number of $|\phi_i\rangle$ given out is equal to the difference in the net numbers of $|\psi_i\rangle$ and $|\psi_{i+1}\rangle$ given out.
\begin{proof}Initially all counts are 0 so the lemma is trivially true. In any query where the first register is 0, the difference between $\Cs_1$ and $\Ds_1$ is preserved (since both are increased or decreased or preserved together) and all other counts are kept the same. Thus, the relations between the counts are preserved. For any query where the first register is $i\in[1,t-1]$, the count in $\Cs_i$ may be decreased, therefore increasing the difference between $\Ds_i$ and $\Cs_i$, but in this case $\Ds_{i+1}$ is increased; no other registers are effected. Thus the relations between the counts are preserved. For any query where the first register is $t$, the counts, and therefore the relations between them, are preserved.
\end{proof}

\subsection{Simulating $\Om_{\Psi,\out}$: State Swap}\label{sec:swap}

Fix a list of states $\Psi$. Now we replace the register $|\cv\rangle$ with the following. Let $\Hs=\Hs_2^n\setminus\{|0\rangle\}$, the space of an $n$-qubit system with the state $|0\rangle$ removed. Recall that $\Sym^\ell\Hs$ is the symmetric subspace of $\Hs^\ell$. Let $\Sym^*\Hs=\cup_{\ell=1}^\infty \Sym^\ell\Hs$.

For $i=1,\dots,t$, let $\Ss_i$ be a copy of $\Sym^*\Hs$; for $i=1,\dots,t-1$, let $\Ts_i$ be another copy of $\Sym^*\Hs$. Let $\Ss$ be the joint system of all $\Ss_i,\Ts_i$. Each $\Ss_i$ is initialized with $\ell$ copies of $|\psi_i\rangle$, and each $\Ts_i$ is initialized with $\ell$ copies of $|\phi_i\rangle$. Here, $\ell\geq T$ is a parameter to be chosen later; think of $\ell$ as polynomial in $T$.

Let ${\sf Incr}_{\Ss_i}$ increase the number of copies of $|\psi_i\rangle$ in $\Ss_i$ by 1 (mod $N$ for some $N> T+\ell$), and likewise define ${\sf Decr}_{\Ss_i},{\sf Incr}_{\Ts_i},{\sf Decr}_{\Ts_i}$. Note that because each $\Ss_i,\Ts_i$ contains many copies of an identical state, the state of the system is always in a symmetric subspace.

Now define the following unitary $\Om^{(2)}_{\Psi,\out,\ell}$ that acts on $\Hs_2^{\otimes(m+n+n)}\otimes \Ss$:
\[\arraycolsep=0pt\begin{array}{lllllllll}
	\Om^{(2)}_{\Psi,\out,\ell}|0\rangle&|0^n\rangle&|0^n\rangle&|\omega\rangle&=|0\rangle&|\psi_1\rangle&|\phi_1\rangle&{\sf Decr}_{\Ss_1}{\sf Decr}_{\Ts_1}|\omega\rangle\\\vspace{5pt} \Om^{(2)}_{\Psi,\out,\ell}|0\rangle&|\psi_1\rangle&|\phi_1\rangle&|\omega\rangle&=|0\rangle&|0^n\rangle&|0^n\rangle&{\sf Incr}_{\Ss_1}{\sf Incr}_{\Ts_1}|\omega\rangle\\
	\Om^{(2)}_{\Psi,\out,\ell}|i\rangle&|\psi_i\rangle&|0^n\rangle&|\omega\rangle&=|i\rangle&|\psi_{i+1}\rangle&|\phi_{i+1}\rangle&{\sf Incr}_{\Cs_{i}}{\sf Decr}_{\Ss_{i+1}}{\sf Decr}_{\Ts_{i+1}}|\omega\rangle&\text{ for }i\in[1,t-1]\\\vspace{5pt}
	\Om^{(2)}_{\Psi,\out,\ell}|i\rangle&|\psi_{i+1}\rangle&|\phi_{i+1}\rangle&|\omega\rangle&=|i\rangle&|\psi_i\rangle&|0^n\rangle&{\sf Decr}_{\Cs_{i}}{\sf Incr}_{\Ss_{i+1}}{\sf Incr}_{\Ts_{i+1}}|\omega\rangle&\text{ for }i\in[1,t-1]\\
	\Om^{(2)}_{\Psi,\out,\ell}|t\rangle&|\psi_t\rangle&|z\rangle&|\omega\rangle&=|t\rangle&|\psi_t\rangle&|z\oplus \out 0^{n-1}\rangle&|\omega\rangle&\text{ for }z\in\{0,1\}^n
\end{array}\]

Observe that, as long as the number of queries is at most $\ell$, then $\Om^{(2)}_{\Psi,\out,\ell}$ can be easily simulated from just the initial state of $\Ss$ containing $\ell$ copies of each of the $|\psi_i\rangle,|\phi_i\rangle$, as well as oracle access to the reflections $P_{|\psi_i\rangle}$ and $P_{|\phi_i\rangle}$ where $P_{|\psi\rangle} := 1- 2|\psi\rangle\langle\psi|$. Indeed, we need two queries to each $P_{|\psi_i\rangle}$ and $P_{|\phi_i\rangle}$ in order to decide if the input register is in one of the states $|\psi_i\rangle,|\phi_i\rangle$, and then uncompute the decision at the end of simulating the query. Moreover, whenever we need to remove an $|\psi_i\rangle$ or $|\phi_i\rangle$ from $|\omega\rangle$, we also need to output an $|\psi_i\rangle$ or $|\phi_i\rangle$, respectively. So instead of deleting, say, one of the copies of $|\psi_i\rangle$ from $|\omega\rangle$, we just put it into the response register given back to the algorithm. Likewise, when we need to increase the number of $|\psi_i\rangle$, we also are given one of the $|\psi_i\rangle$ as input. Since the input $|\psi_i\rangle$ needs to be deleted to execute the gate, we can instead just swap the $|\psi_i\rangle$ given as input into $|\omega\rangle$, simultaneously deleting the input copy and increasing the number of copies in $|\omega\rangle$, as desired. The only issue is if the number of copies drops below 0 or increases to $N$ or larger, in which case the number of copies gets reduced mod $N$. But since we started with $\ell$ copies which is at least the number of queries, then we can never run out of copies. Likewise, the number of copies can never increase by more than $T$, for a total of $\ell+T<N$. Thus, we never need to reduce the number of copies $\bmod N$. See Figure~\ref{fig:swap}.

\begin{figure}[htb]
	\centering
	\includegraphics[width=.8\textwidth]{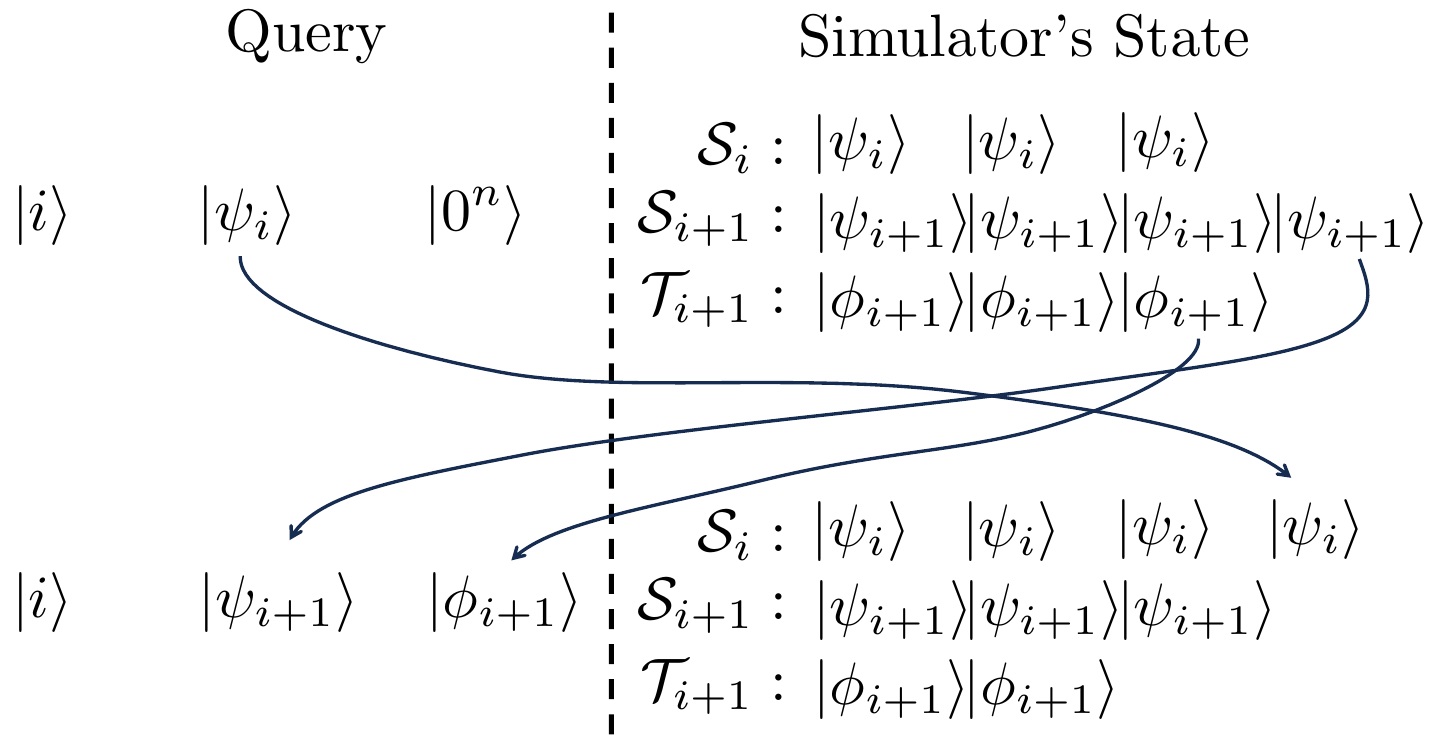}
	\captionsetup{width=.8\textwidth}
	\caption{\label{fig:swap} How our simulator maps the query input (top) to the query output (bottom) by simply moving registers around.}
\end{figure}

We also have the following:

\begin{lemma}\label{lem:swap} Let $\As$ be a time $T\leq \ell$ algorithm. Then for any $\Psi$, any $\out$, we have the following equality of density matrices:
	\[\As^{\Om^{(1)}_{\Psi,\out}}()=\As^{\Om^{(2)}_{\Psi,\out,\ell}}()\]
\end{lemma}
\begin{proof}We can compute $|\omega\rangle$ from $|\cv\rangle$ (assuming knowledge of the $|\psi_i\rangle,|\phi_i\rangle$), and vice versa, as follows. The count in $\Cs_i$ is just $\ell$ minus the number of copies of $|\psi_i\rangle$ in $\Ss_i$. Likewise the count in $\Ds_i$ is just $\ell$ minus the number of copies of $|\phi_i\rangle$ in $\Ts_i$. Therefore, since $|\omega\rangle$ can be computed from $|\cv\rangle$ just by computing on registers of the simulator, the algorithm cannot distinguish whether $|\omega\rangle$ or $|\cv\rangle$ is stored by the simulator. The only issue is if the number of copies of some state in $|\omega\rangle$ gets reduced mod $N$, but this cannot happen by our choice of $\ell\geq T$ and $N\geq T+\ell$.
\end{proof}
Let $c_i$ be $\ell$ minus the number of copies of $|\phi_i\rangle$ and $d_i$ be $\ell$ minus the number of copies of $|\psi_i\rangle$ in $|\omega\rangle$. By mapping the constraints on $\cv$ from Lemma~\ref{lem:countconstraints} to $|\omega\rangle$, we also have:
\begin{corollary}\label{cor:countconstraints} At any point when running $\As^{\Om^{(2)}_{\Psi,\out,\ell}}()$, the support of the simulator's state is only on terms satisfying, $d_{i+1}=d_i-c_i$.
\end{corollary}

\subsection{Simulating $\Om_{\Psi,\out}$: Approximating $|\psi\rangle\langle\psi|$}\label{sec:proj}

Above in Section~\ref{sec:swap}, we show how to \emph{almost} simulate $\Om_{\Psi,\out}$ just using copies of $|\psi_i\rangle$ and $|\phi_i\rangle$. The only part where we need actual knowledge of $|\psi_i\rangle,|\phi_i\rangle$ is to implement the reflections $P_{|\psi_i\rangle},P_{|\phi_i\rangle}$. Here, we use techniques from~\cite{C:JiLiuSon18} to simulate queries to the reflection, just using our copies of $|\psi_i\rangle$.

Let $\As^{P_{|\psi\rangle}}$ be an algorithm making $Q$ queries to $P_{|\psi\rangle}$.~\cite{C:JiLiuSon18} simulate the queries to $P_{|\psi\rangle}$ as follows. Initialize a register $\Sym^\ell\Hs$ to contain $\ell$ copies of $|\psi\rangle$. Now, instead of responding to each query with $P_{|\psi\rangle}$, respond to each query with $\Sym^{\ell+1}$, the reflection about the symmetric subspace of $\ell+1$ copies of $\Hs$, where $\ell$ copies come from the simulator's register $\Sym^\ell\Hs$, and the remaining register is the query.

Let $\rho_0$ be the final state of the algorithm $\As^{P_{|\psi\rangle}}$ when making queries to the actual reflection, together with $\ell$ copies of $|\psi\rangle$. Let $\rho_1$ be the final state of $\As$ when the queries are simulated, together with the final state of $\Sym^\ell\Hs$ (which is symmetric but may no longer be identical copies of $|\psi\rangle$ since the simulation will have perturbed them).

\begin{lemma}[\cite{C:JiLiuSon18}, Theorem 4]\label{lem:projprior} $TD\left[\rho_0,\rho_1\right] \leq \frac{2Q}{\sqrt{\ell+1}}$
\end{lemma}

In our case, we have to be a bit careful applying Lemma~\ref{lem:projprior}, since we do not have a fixed number of copies of $|\psi\rangle$, and the states in $\Ss_i,\Ts_i$ can be in superposition of having differing numbers of copies. Instead, we will need the following refinement. Initialize a register $\Sym^*\Hs$ to contain $\ell$ copies of $|\psi\rangle$. Now, respond to each query with the reflection $\Sym^{*+1}$: for states in $\Sym^*\Hs$ contained in $\Sym^{\ell'}\Hs$, $\Sym^{*+1}$ will reflect about the symmetric subspace of the joint system of $\Sym^{\ell'}\Hs$ and the query register. Between queries, $\As$ is now allowed to add or remove copies of $|\psi\rangle$ from $\Sym^*\Hs$. Let $T$ be an upper bound on the number of copies that can be removed. Let $\rho_0$ be the final joint state $\As$ and $\Sym^*\Hs$ when $\As$'s queries are answered by $P_{|\psi\rangle}$, and let $\rho_0$ be the final joint state when the queries are answered by $\Sym^{*+1}$.

\begin{corollary}\label{lem:projpriorcor} If the number of removed copies is at most $T$, $TD\left[\rho_0,\rho_1\right] \leq \frac{2Q}{\sqrt{\ell-T+1}}$
\end{corollary}
\begin{proof}This follows from a simple hybrid argument. Let $H_0$ be the case where $\As$'s queries are answered with $P_{|\psi\rangle}$, and $H_i$ be the case where the first $Q-i$ queries are answered with $P_{|\psi\rangle}$, and the remaining queries are answered with $\Sym^{*+1}$. It suffices to prove that the trace distance between $H_i$ and $H_{i+1}$ is at most $2/\sqrt{\ell-T+1}$, and the triangle inequality implies the lemma.
	
Toward that end, observe that $H_i,H_{i+1}$ are identical except for the $i$th query from the end. Up until this point, $\Sym^*\Hs$ has not been used to answer queries, though it may have had some copies of $|\psi\rangle$ added or removed. Therefore, the state of $\Sym^*\Hs$ is a superposition over $|\psi\rangle^{\otimes \ell'}$ for several different $\ell'$. Since $\As$ is only allowed to remove up to $T$ of the copies, we know that the support of this state has $\ell'\geq \ell-T$. It is therefore a straightforward application of Lemma~\ref{lem:projprior} that the trace distance between $H_i,H_{i+1}$ is at most $\frac{2}{\sqrt{\ell-T+1}}$, as desired.
\end{proof}

We now apply Lemma~\ref{lem:projprior} to for each $|\psi_i\rangle,|\phi_i\rangle$. We set $Q=2T$ and use that our simulation $\Om^{(2)}$ makes $2T$ queries to each projection oracle (2 for each of $\As$'s $T$ queries) and removes at most $T$ copies of each $|\psi_i\rangle,|\phi_i\rangle$. We therefore obtain a simulator $\Om^{(3)}_{\Psi,\out,\ell}$ which is given $\ell$ copies of each of the $|\psi_i\rangle,|\phi_i\rangle$, and attempts to simulate $\Om^{(2)}_{\Psi,\out,\ell}$. We immediately have:

\begin{lemma}\label{lem:swap2} Let $\As$ be a time $T$ algorithm. Then for any $\ell>T$, any $\Psi$, any $\out$, and output $z$,
	\[TD\left[\As^{\Om^{(2)}_{\Psi,\out,\ell}}(),\As^{\Om^{(3)}_{\Psi,\out,\ell}}()\right]\leq \frac{8tT}{\sqrt{\ell-T+1}}\]
	where the states on both sides include the register provided to $\Om^{(2)},\Om^{(3)}$ which initially contains the $\ell$ copies of each of the $|\psi_i\rangle,|\phi_i\rangle$.
\end{lemma}

Now measure the number of $\Hs_2^n$ registers in each of $\Ss_i$ and $\Ts_i$, obtaining values $\ell-c_i$ and $\ell-d_i$ for integers $c_i,d_i$. By Corollary~\ref{cor:countconstraints}, we have:
\begin{corollary}\label{cor:countconstraints2} With probability 1, $d_{i+1}=d_i-c_i$.
\end{corollary}

\begin{lemma}Except with probability at most $\frac{8tT}{\sqrt{\ell-T+1}}+\frac{2t\times\ell}{2^n-1+\ell}$, $c_i\geq 0$ and $d_i\geq 0$ for all $i$.
\end{lemma}
\begin{proof}If any $c_i$ (resp. $d_i$) are less than zero, it means the number of ``copies'' of $|\psi_i\rangle$ of the simulator exceeds the original number provided originally. If these were actually perfect copies, then this would violate the unclonability of Haar random states. Indeed, it is known~\cite{Werner98} that for a Haar random state over dimension $D$, the probability of mapping $r$ copies to $r+1$ is bounded by $\ell/(D+\ell)$. In our case, $D=2^n-1$ (since the states are Haar random in $\Hs$, which is $\Hs_2^{\otimes n}$ excluding $|0\rangle$). Then we can union bound over all $2t-1\leq 2t$ states $|\psi_i\rangle,|\phi_i\rangle$, to get the probability  of any $c_i$ or $d_i$ being less than 0 being at most $2t\times\ell/(2^n-1+\ell)$.

Now, the copies provided to the simulator have potentially been perturbed as the simulator runs. However, Lemma~\ref{lem:swap2} implies that they can only have been perturbed by $\frac{8tT}{\sqrt{\ell-T+1}}$, meaning the states are still close to the respective $|\psi_i\rangle,|\phi_i\rangle$. Putting these together completes the proof of the lemma.\end{proof}

\subsection{Putting it All Together}\label{sec:final}

When $\As$ terminates, apply the operation $P_{|\psi_t\rangle}$ to the output. If it accepts, then add the resulting state to $\Ss_t$. By piecing together the above results, we therefore have an algorithm which, with probability at least
\[W:=\frac{1}{(48T)^2}-\frac{8tT}{\sqrt{\ell-T+1}}-\frac{2t\times\ell}{2^n-1+\ell}\]
results in $c_t\geq 1$, and for $i\in[1,t-1]$, $c_i,d_i\geq 0$ and $d_{i+1}=d_i-c_i$. If we assume $T\geq \max(n,t)$ and let $\ell=\Omega(T^6)=\Omega(t^2 T^4)$, and if we assume $T^6\ll 2^n$, we can lower bound $W$ as $\Omega(T^{-2})$.

But observe that in this case, we must have all $d_i\geq 1$. In this case, the system has collapsed to a space of lower dimension. Specifically, as each of the registers $\Fs_i,\Gs_i$ are in the symmetric spaces $\Sym^{\ell-c_i}\Hs,\Sym^{\ell-d_i}\Hs$, their dimension is $\binom{(2^n-1)+(\ell-c_i)-1}{\ell-c_i},\binom{(2^n-1)+(\ell-d_i)-1}{\ell-d_i}$, respectively. Thus, if we let $S$ be the algorithm's space, and using that the $d_i\geq 1$ and the $c_i\geq 0$, the total dimension of the joint system of the simulator's state and algorithm's state is at most \[D_{\sf Final}:=\binom{(2^n-1)+\ell-1}{\ell}^t\times\binom{(2^n-1)+(\ell-1)-1}{\ell-1}^{t-1}\times 2^{S}\]

On the other hand, these spaces all started in the symmetric subspace $\Sym^{\ell}\Hs$, which has dimension $\binom{(2^n-1)+\ell-1}{\ell}$. Specifically, since the $|\psi_i\rangle$ and $|\phi_i\rangle$ are Haar random, the initial mixed state is equivalent to the totally mixed state in this symmetric subspace. The algorithm's state starts out deterministically in the state $|0\rangle$. Thus, the initial state joint state of the algorithm and simulator is a totally mixed state in a space of dimension
\[D_{\sf Initial}:=\binom{(2^n-1)+\ell-1}{\ell}^{2t-1}\times 1\]

\begin{lemma}\label{lem:rank}Let $\rho$ be a totally mixed state in a subspace of dimension $D_{\sf Initial}$. Let $U$ be a unitary. Let $S_{\sf Final}$ be any subspace of dimension $D_{\sf Final}$, which we will also associate with the projection onto that space. Then $\Tr[S_{\sf Final}U\rho U^\dagger S_{\sf Final}]\leq D_{\sf Final}/D_{\sf Initial}$. In other words, the probability that a totally mixed state in dimension $D_{\sf Initial}$ can be mapped to a space of dimension $D_{\sf Final}$ using unitary computations is at most $D_{\sf Final}/D_{\sf Initial}$.
\end{lemma}
\begin{proof}Since $\rho$ is a totally mixed state in a subspace of dimension $D_{\sf Initial}$, it has $D_{\sf Initial}$ positive eigenvalues, all equal to $1/D_{\sf Inital}$. On the other hand, the state $S_{\sf Final}U\rho U^\dagger S_{\sf Final}$ has rank at most $D_{\sf Final}$, and therefore the number of non-negative eigenvalues is at most $D_{\sf Final}$. Moreover, if $\lambda$ is the maximal eigenvalue of $S_{\sf Final}U\rho U^\dagger S_{\sf Final}$ and $|\tau\rangle$ the associated maximal eigenvector, then
	\[\lambda=\langle \tau |S_{\sf Final}U\rho U^\dagger S_{\sf Final}|\tau\rangle=\langle\tau'|\rho|\tau'\rangle\leq (1/D_{\sf Inital})\langle\tau'|\tau'\rangle\leq 1/D_{\sf Inital} \]
	where $|\tau'\rangle=U^\dagger S_{\sf Final}|\tau\rangle$, which has norm at most 1 since it is the projection of a norm-1 vector. In other words, the maximal eigenvalue $\lambda$ is at most the maximal eigenvalue of $\rho$. Since the number of non-negative eigevalues is at most $D_{\sf Final}$, the trace, which equals the sum of all eigenvalues, is at most $D_{\sf Final}/D_{\sf Inital}$, as desired.
\end{proof}

Applying Lemma~\ref{lem:rank}, we therefore have:
\[\Omega(T^{-2})\leq D_{\sf Final}/D_{\sf Initial}= \left(\frac{\ell}{(2^n-1)+\ell}\right)^{t-1}\times 2^{S}\]
Rearranging and taking logarithms gives
\[S\geq \Omega(\;t n-t\log \ell-\log T\;)\]
Using our assumption that $T^6\ll 2^n$ (equivalently, $T\ll 2^{n/6}$) and setting $\ell\geq\Omega(T^6)$, we have that $S\geq \Omega(tn)$. This completes the proof of Lemma~\ref{lem:unitarycircuitlower}.

\section{Our Main Theorem}\label{sec:main}
We now prove our main theorem.

\begin{theorem}\label{thm:main} Fix a proper universal measurement set. Then there are constants $c,d$ such that there is no black box $(S'=cT,T'=2^{dS})$-space-time purifier.
\end{theorem}

In other words, either $S'\geq \Omega(T)$ or $T'\geq 2^{\Omega(S)}$. The rest of this section will be devoted to proving Theorem~\ref{thm:main}. 

\paragraph{Roadmap.} We break into three cases: the first is when  $T$ is small, specifically $T<S$. In this case, we can trivially bound $S'$. The second case is when $T$ is large, specifically $T\geq 2^{\Omega(S)}$. In this case, we can bound $T'$ easily (regardless of $S'$) using known query complexity lower-bounds. In both these ``easy'' cases, we actually do not even care that the compiler is removing measurements, as the bounds hold for \emph{any} circuit that approximates the starting circuit. In the final and most interesting case when $S\leq T\leq 2^{O(S)}$. Here is where we finally use the fact that the compiler is actually removing measurements, and invoke our oracle from Section~\ref{sec:sep}. This requires some care, since our model allows the compiled program to effectively have non-uniform advice about the oracle, which our analysis in Section~\ref{sec:sep} did not handle. We then use known techniques to get a lower-bound even in the case of advice, thus completing the theorem.

\subsection{The case $T<S$}

Let $\Gs_0$ be the proper universal gate set. Let $\Um$ be an arbitrary permutation matrix acting on $S$ qubits. Consider the gate set $\Gs_0\cup\{\Um\}$. Consider the space-$S$, time-$T$ for $T=1<S$ circuit $C$ which simply queries $\Um$ on its input. Now consider any other circuit $C'$ which approximates $C$, meaning it computes the permutation $\Um$. Since $C'$ must act on input and output of $S$ qubits, $C'$ must have space $S'\geq S$. Thus no compiler (in particular, no purifier) can map space $S$ to space $S'<S$. In particular, if $T<S$,  then any compiler/purifier must have $S'>T$.

\subsection{The case $T\geq 2^{\Omega(S)}$}

Let $\Gs_0$ be the proper universal gate set. Let $H:\{0,1\}^{n}\rightarrow\{0,1\}^{n}$ be a permutation, and consider the gate set $\Gs_0\cup\{\Hm\}$, where $\Hm$ is the unitary representation of $H$, which acts on $2n$ qubits. We first recall the foundational Grover's algorithm, re-expressed in our language.

\begin{lemma}[\cite{STOC:Grover96}]\label{lem:grover} There is a quantum circuit $C$ running in space $S=2n$ and time $O(2^{n/2})=2^{O(S)}$ over $\Gs_0\cup\{\Hm\}$, such that for any $x$, $\Pr[C(x)=H^{-1}(x)]\geq 1-O(2^{-n})$.\end{lemma}

On the other hand, we know by the optimality of Grover search~\cite{BBBV97}, that \emph{any} circuit must make $\Omega(2^{n/2})$ queries to $\Hm$, and must therefore run in at least as much time. However, we note that our notion of compiler allows the compiled circuit to depend on the gate set, and in particular the function $H$. But a priori it may be possible for a circuit that depends on $H$ to beat the basic $2^{n/2}$ bound. Fortunately, this is the domain of pre-processing attacks, and is well-understood. We recall the following theorem of Nayebi et al.~\cite{QIC:NABT15}, again rephrased in our language:
\begin{theorem}[\cite{QIC:NABT15}] There is a polylogarithmic function $\polylog$ such that the following is true. If $H$ is a random permutation, then for any circuit $C'$ that may depend on $H$, if $C'$ has at most $Q$ $\Hm$ gates and total size $T'$ such that $Q^2T'\polylog(T')\leq 2^n$, then $\Pr[C'(x)=H^{-1}(x)]< 2/3$.
\end{theorem}

Note that \cite{QIC:NABT15} considers quantum advice, as opposed to circuits that depend on $H$. But these are equivalent by thinking of the circuit of size $T'$ itself as the advice, which takes $T'\log T'$ bits to write down.

Since $Q\leq T'$, we therefore have that for $C'$ to approximate $C$, it must be that $T'\geq \Omega(2^{n/3}/\poly(n))\geq 2^{\Omega(n)}=2^{\Omega(S)}$. Thus, for any $S$ and any $T\geq 2^{\Omega(S)}$, we can choose $n$ to be a sufficiently small constant multiple of $S$ so that Grover's algorithm runs in time $\leq T$ and space $\leq S$. Meanwhile, any compiler (whether or not it is a purifier) must have $T'\geq 2^{\Omega(S)}$.

\subsection{The case $S\leq T\leq 2^{O(S)}$}

Let $t=\Theta(T/S)$ and $n=\Theta(S)$ and $m=\lceil\log t+1\rceil$. We construct a unitary $\Um$ as follows. Choose a random function $F:\{0,1\}^n\rightarrow\{0,1\}$. For each $x\in\{0,1\}^n$, sample $\Psi_x$ as the description of $2t-1$ Haar random states. Then $\Um$ is the oracle acting on $\Hs_2^{\otimes(n+m+n+n)}$ defined as:
\[\Um=\sum_{x\in\{0,1\}^n}|x\rangle\langle x|\otimes \Om_{\Psi_x,F(x)}\]
	
Our gate set will be $\Gs_0\cup\{\Um\}$, where $\Gs_0$ is the fixed proper universal general gate set. 

Observe that we can use $\Um$ and the algorithm from Lemma~\ref{lem:generalcircuitupper} to evaluate the function $F(x)$ with perfect probability, in time $O(nt)=T$ and space $O(n+\log(t))=O(n)=S$, using measurements. It remains to prove that no \emph{unitary} algorithm running in space $o(T)$ and time $2^{o(S)}$ with gate set $\Gs_0\cup\{\Um\}$ can even approximately compute $F(x)$. To do so, we use the lower bound from Lemma~\ref{lem:unitarycircuitlower}. That result showed a lower bound for low space-time algorithms to compute the bits $F(x)$ for any $x$. However, the lower bound in Lemma~\ref{lem:unitarycircuitlower} only applies to circuits that are specified independently from the choice of unitary $\Um$. Here, we show that even if the choice of circuit depends on $\Um$, then it still remains hard to compute $F(x)$ on most $x$. The idea is to view the circuit description as non-uniform advice about $\Um$. This advice could have some of the $F(x)$ hardcoded. However, because the circuit needs to be small, the advice is therefore small, and by the incompressibility of random strings, most of $F(x)$ cannot be hardcoded. The challenge here is that there may be a clever algorithm which takes the advice string and also makes queries to $\Um$, and is able to use the advice to reduce the space-time requirements for computing $F(x)$, even if $F(x)$ is not directly hardcoded in the advice. Here, we show that this is not possible.

Note that the situation is similar to the ``salting defeats pre-processing'' setting of~\cite{FOCS:CGLQ20}. There, the authors consider a general cryptographic game relative to a random classical oracle, and show that if a game is hard relative to algorithms that are independent of the random oracle, then it is also hard relative to algorithms with advice, provided the game is ``salted,'' meaning a random salt of sufficient length is appended to the random oracle queries. In our setting, the original un-salted game is to compute $o$ given $\Om_{\Psi,o}$, and the salted game is to compute $F(x)$ for a random $x$, given $\Um$. Unfortunately, it is unclear how to apply their result to our setting, since our oracle $\Om$ is a structured quantum oracle and not a random classical oracle. We also care about the space complexity of the game, which is not considered in~\cite{FOCS:CGLQ20}.

Instead, we will borrow a technique from~\cite{FOCS:YamZha22}, which can also be seen as a ``salting defeats pre-processing'' result, except it is far more flexible. On the other hand, the quantitative bounds are weaker, though they are sufficient for our purposes.

\paragraph{A Technical Lemma.} Here we give a simple lemma which states that, given joint leakage on many iid random variables sampled from a distribution $\Ds$, the marginal distribution of most of the variables will still be distributed approximately as $\Ds$.

\begin{lemma}\label{lem:stats} Let $\Ds$ be a distribution and $X_1,\dots,X_g,Y$ be iid random variables sampled from $\Ds$. Let $L$ be a function with co-domain of size $2^r$. Then 
\[\Delta(\;(i,X_i,F(X_1,\dots,X_g))\;,\;(i,Y,L(X_1,\dots,X_g))\;)\leq \sqrt{r/2g}\]
Above, $i$ is uniform in $[g]$.
\end{lemma}
\begin{proof}Let $I(X;Y)$ denote the mutual information between random variables $X$ and $Y$. Then
	\[r\geq I(\;L(X_1,\dots,X_g)\;;\;X_1,\dots,X_t\;)\geq \sum_{i=1}^g I(\;L(X_1,\dots,X_g)\;;\; X_i\;)\]
where the second inequality is due to the independence of the $X_i$. Let $\delta_i$ be the statistical distance between $(L(X_1,\dots,X_g),X_i)$ and $(L(X_1,\dots,X_g),Y)$. Let $\delta$ be the statistical distance between $(i,X_i,L(X_1,\dots,X_g))$ and $(i,Y,L(X_1,\dots,X_g))$; our goal is to bound $\delta$. $I(\;L(X_1,\dots,X_g)\; X_i\;)$ is just the KL divergence between $(L(X_1,\dots,X_g),X_i)$ and $(L(X_1,\dots,X_g),Y)$. By Pinsker's inequality, we therefore have that $\delta_i\leq \sqrt{I(\;L(X_1,\dots,X_t)\; X_i\;)/2}$. This implies
\[r\geq 2\sum_{i=1}^g \delta_i^2\]
On the other hand, $\delta=(\sum_i \delta_i)/g$. Jensen's inequality then gives that
\[\delta\leq \sqrt{\sum_i \delta_i^2/g}\leq \sqrt{r/2g}\qedhere\]	
\end{proof}

\paragraph{Our Reduction.} We are now ready to give our reduction.
\begin{lemma}\label{lem:reduction}Consider a function $\Cs(F,\{\Psi_x\}_x)$ which outputs a circuit $\As$ of size $r$, where $\As$ makes queries to $\Um$ and runs in space $S$ and time $T$. Let $T\leq 2^{n/4}$. Suppose that with probability $p$ over the choice of $F$ and $\{\Psi_x\}_x$, $\As$ is able to compute $F(x)$ with probability $2/3$ for all $x$. Then there is a unitary circuit $\Bs$ which gets no input and makes queries to $\Om_{\Psi^*,o^*}$ for a random choice of $\Psi^*,o^*$, and is able to output $o^*$ with probability at least $7/12$ for sufficiently large $n$. $\Bs$ makes at most $O(T)$ queries to $\Om_{\Psi^*,o^*}$ and for an appropriate gate set (which is independent of $\Om_{\Psi^*,o^*}$) has space $O(S)$.
\end{lemma}
\begin{proof}To construct $\Bs$, do the following. First choose a random subset $V\subseteq\{0,1\}^n$ of size $g$, for a parameter $g$ to be chosen later. Also choose a random $x^*\in\{0,1\}^n\setminus V$. Then choose $F,\{\Psi_x\}_x$, and let $\As=\Cs(F,\{\Psi_x\}_x)$. Now, for each $x\in V$, sample $o_x'\gets\{0,1\}$ and a fresh random $\Psi_x'$. We define a unitary $\Vm$ as
	
\[\Vm=\sum_{x\in\{0,1\}^n\setminus V}|x\rangle\langle x|\otimes \Om_{\Psi_x,F(x)}+\sum_{x\in V}|x\rangle\langle x|\otimes \Om_{\Psi_{x}',o_{x}'}\]
	
$\Bs$ will run $\As$ on input $x^*$, but answer $\As$'s queries with the unitary $\Um'$ defined as:
\[\Um'=\sum_{x\in\{0,1\}^n\setminus (V\cup\{x^*\})}|x\rangle\langle x|\otimes \Om_{\Psi_x,F(x)}+\sum_{x\in V}|x\rangle\langle x|\otimes \Om_{\Psi_{x}',o_{x}'}+|x^*\rangle\langle x^*|\otimes \Om_{\Psi^*,o^*}\]
$\Bs$ can simulate $\Um'$ with minimal time and space overhead by making a single query each to $\Vm$ and $\Om_{\Psi^*,o^*}$. Thus, by including $\Vm$ in the gate set for $\Bs$, the time and space complexity are linear in that of $\As$.

We lower bound $\Bs$'s success probability in a few steps.

\begin{claim}Fix $\Um,\Am,x^*$. Then $|\Pr[\As^\Um(x^*)=F(x^*)]-\Pr[\As^\Vm(x^*)=F(x^*)]|\leq 4T\sqrt{g/(2^n-1)}$, where the probability is taken over the randomness of $\Am$ and  the choice of $\Vm$.
\end{claim}
\begin{proof}$\Um$ and $\Vm$ differ only on inputs $x\in V$, which is chosen randomly. Consider running $\As^\Um(x^*)$. Since the view of $\As$ is independent of $V$ (except that $x^*\notin V$), the expected query weight of points in $V$ is at most $\epsilon=gT/(2^n-1)$, with equality obtained only if the query weight on $x^*$ is 0. Lemmas~\ref{lem:bbbv1} and~\ref{lem:bbbv2} then imply the difference in output probabilities is at most $4\sqrt{T\epsilon}$, proving the claim.
\end{proof}
\begin{claim}Let $r$ be the description size of $\As$. Then $|\Pr[\As^\Vm(x^*)=F(x^*)]-\Pr[\As^{\Um'}(x^*)=F(x^*)]|\leq \sqrt{r/2g}$, where the probability is taken over the choice of $\Um,$V$,\Vm,\Am,x^*$ and the randomness from running $\Am$.
\end{claim}
\begin{proof}We apply Lemma~\ref{lem:stats}, where $X_1,\dots,X_{g+1}$ are the variables $(\Psi_x,F(x))$ for $x\in V\cup\{x^*\}$, and $L$ is the function obtained from $\Cs$ by fixing $(\Psi_x,F(x))$ for $x\notin V\cup\{x^*\}$. Since $x^*$ is uniform in $V\cup\{x^*\}$ and all the $(\Psi_x,F(x))$ for $x\in V$ have been replaced with fresh random samples, Lemma~\ref{lem:stats} shows that $(\Psi_{x^*},F(x^*))$ and a fresh uniform $(\Psi^*,o^*)$ have statistical distance $\sqrt{r/2g}$, thus bounding the change in success probability.
\end{proof}

Now we set $g=\sqrt{(2^n-1)r/32T^2}$, we have that 
\[\Pr[\Bs^{\Om_{\Psi^*,o^*}}()=o^*]=\Pr[\As^{\Um'}(x^*)=F(x^*)]\geq \frac{2}{3}-\sqrt[4]{\frac{32T^2 r}{2^n-1}}\]
Since $r$ is the description size of a circuit of size $T$, we have that $r=T\times\polylog(T)$. Therefore, since $T\leq 2^{n/4}$, for sufficiently large $n$ we have that $\sqrt[4]{\frac{32T^2 r}{2^n-1}}\leq 1/12$. Thus $\Pr[\Bs^{\Om_{\Psi^*,o^*}}()=o^*]\geq 7/12$.\end{proof}

Combining Lemma~\ref{lem:reduction} with Lemma~\ref{lem:unitarycircuitlower} (which rules out such a $\Bs$) finishes the case $S\leq T\leq 2^{O(S)}$. This finishes the proof of Theorem~\ref{thm:main}.\qed

\bibliographystyle{alpha}
\bibliography{abbrev3,crypto,bib}

\end{document}